\documentclass[a4]{article}

\usepackage{hyperref,a4wide}
\usepackage{pgfarrows,tikz}
\usetikzlibrary{shapes,arrows}
\usetikzlibrary{automata}
\usetikzlibrary{positioning}
\usepackage[utf8]{inputenc}
\usepackage[T1]{fontenc}
\usepackage{amsthm,amsmath,amsfonts, amssymb,amscd,mathbbol,graphicx}
\newcommand{\Z}{\mathbb{Z}}
\newcommand{\N}{\mathbb{N}}

\newtheorem{theorem}{Theorem}
\newtheorem{corollary}{Corollary}

\newtheorem{lemma}[theorem]{Lemma}

\theoremstyle{definition}
\newtheorem{definition}[theorem]{Definition}
\newtheorem{example}{Example}
\newtheorem{remark}{Remark}

\newcommand\graphs[1]{\mathcal{G}(#1)} 
\newcommand\cas[1]{\mathcal{CA}(#1)} 
 
\newcommand\isbydef{\overset{\textrm{def}}{=}}
\newcommand\quantif{\textbf{Q}}

\title{FO logic on cellular automata orbits equals MSO logic}

\author{Guillaume Theyssier (I2M, CNRS, Université Aix-Marseille, France)}

\begin{document}
\maketitle

\begin{abstract}
  We introduce an extension of classical cellular automata (CA) to arbitrary labeled graphs, and show that FO logic on CA orbits is equivalent to MSO logic. 
  We deduce various results from that equivalence, including a characterization of finitely generated groups on which FO model checking for CA orbits is undecidable, and undecidability of satisfiability of a fixed FO property for CA over finite graphs. We also show concrete examples of FO formulas for CA orbits whose model checking problem is equivalent to the domino problem, or its seeded or recurring variants respectively, on any finitely generated group. For the recurring domino problem, we use an extension of the FO signature by a relation found in the well-known Garden of Eden theorem, but we also show a concrete FO formula without the extension and with one quantifier alternation whose model checking problem does not belong to the arithmetical hierarchy on group $\Z^2$.
\end{abstract}

\section{Introduction}
\label{sec:intro}

Symbolic dynamics was historically introduced as the study of one-dimensional infinite words representing discretized orbits of smooth dynamical systems through a finite partition of space \cite{morsehedlund,Lind_2020}. It has since then been largely extended to higher dimensions and arbitrary Cayley graphs of finitely generated groups, and seen rich developments going way beyond the initial motivations. The field of symbolic dynamics would now be better described as the study of sets of configurations (\textit{i.e.} coloring of a graph) which can be defined by local uniform constraints (subshift of finite type, sofic subshift, etc) and maps on configurations acting by a uniform and local update rule (cellular automata, and morphisms between subshifts, which are the continuous maps commuting with translations \cite{hedlund,cagroups}).
 A fascinating aspects of these objects is that they are very simple to define, yet can produce very complex behaviors which make them challenging to analyze. They can be considered as a reasonable modeling tool \cite{Chopard_1998}, but more importantly, they constitute a natural model of computation for which undecidablity and computational hardness can arise in the most simple and seemingly unrelated questions in a spectacular way \cite{entrsft,Kari_1994,Harel_1985,berger,Kari_injsurj}.

 In symbolic dynamics, a major trend is to relate the properties of the considered objects (cellular automata or subshifts) to the structure of the underlying graph they are defined on (usually Cayley graphs of finitely generated groups). An emblematic example is the domino problem: since the breakthrough undecidability result of Berger \cite{berger}, a lot of works have focused on the characterization of Cayley graphs of groups for which the domino problem is undecidable  \cite{Aubrun_2018,dominosurface,Bartholdi_2022,bartholdi2023domino} or other graphs with less symmetries \cite{dominoexcludedminor,Hellouin_de_Menibus_2023}. Concerning cellular automata (CA),
properties of the global map that can be expressed in first-order (FO) logic (or, said differently, FO properties of their orbit graph) are already challenging. 
For instance, injectivity and surjectivity problems were shown decidable on $\Z$ \cite{Amoroso_1972} and on context-free graphs \cite{Muller_1985}, but undecidable on $\Z^2$ \cite{Kari_injsurj}. Other FO properties of CA were studied in relation to the graph structure: 
the Gotschalk conjecture \cite{cagroups,Gottschalk_1973} asks whether the property ``injectivity implies surjectivity'' is true for all CA on all group. Besides, the garden of Eden theorem \cite{cagroups} characterizes amenability among finitely generated groups by the fact that another simple FO property\footnote{As detailed in Section~\ref{sec:cayley-graphs-domino}, it uses an additional relation in the FO signature which doesn't break the main point of our approach, an equivalence with MSO.} is true for all CA on this group.

It turns out that many central problems considered in symbolic dynamics can actually be rephrased in \emph{monadic second order logic} (MSO). It was for instance noticed in \cite{Muller_1985} for injectivity and surjectivity of CA. MSO is a logical language that has received enormous interest, probably for its balance between expressivity and decidability in many cases.
More precisely, for graph properties, algorithmic metatheorems \cite{Muller_1985,Courcelle_1990,courcellebook} and reciprocals \cite{KreutzerT10,Kuske_2005} relates the complexity of the MSO model checking problem to the structure of considered graph or graph family. The key parameter at play here is treewidth \cite{Robertson_treewidth} which is related to the non-existence of arbitrarily large grid minors \cite{Robertson_1986}: bounded treewidth provides positive algorithmic results, while arbitrarily large grid minors often allows lower bounds or undecidability results. For instance, on a Cayley graph of a finitely generated group, the MSO model checking problem is decidable if and only if the group is virtually free (meaning having a free group as finite index subgroup) \cite{Kuske_2005}.

Of course, positive algorithmic results for MSO apply to the particular problems studied in symbolic dynamics, but lower bounds or undecidability results for MSO are not directly transferable. For instance, on the graph $\Z^2$, undecidability of MSO model checking follows from a straightforward encoding of Turing machines, while undecidability of the domino problem \cite{berger} or of injectivity problem of CA \cite{Kari_injsurj} requires detailed and non-straightforward constructions that involves ideas and tools of independent interest (aperiodic tile sets and space filling curves for instance). Moreover, it is still open to our knowledge whether both problems are undecidable on any Cayley graph of group which is not virtually free. More generally, much work remains to close the gap between the global understanding of MSO logic on arbitrary graphs and the particular MSO fragments at sake in symbolic dynamics that are mostly understood on Cayley graphs of some group families.

\subparagraph*{Contributions of the paper.}
In this paper we explore properties of CA on arbitrary labeled graphs (finite or infinite). To do this, we introduce a definition of local rules of CA that doesn't use any explicit reference to the local structure of the graph as it is classically done, but instead just relies on the notion of bounded labeled walks and multiset of possible states read at the end of these walks. In particular, we don't require the graph to be uniform nor to have bounded degree, but our notion exactly corresponds to the classical one on Cayley graphs of finitely generated groups. Moreover, a fixed CA local rule provides a well-defined CA global map on any graph with same label sets. We can thus explore the influence of the graph separately from the local rule. For instance, we can specify a graph property by a CA local rule $f$ and a FO property $\phi$ of CA orbits:  the set of graphs on which the local rule $f$ induces a global map that satisfies $\phi$. Our main result (theorems~\ref{theo:FOCAtoMSO}, \ref{theo:MSOFOCACON} and \ref{theo:MSOFOCA}) is then an equivalence between this approach and MSO logic, on arbitrary labeled graphs. Precisely, for every class $\mathcal{C}$ of graphs, the following two conditions are equivalent:

\begin{itemize}
\item there exists an MSO formula $\Psi$, such that $G\in\mathcal{C} \iff G\models\Psi$,  
\item there exists a pair $(\phi,f)$, where $\phi$ is a FO formula and $f$ is a CA local rule, such that ${G\in \mathcal{C}}$ if and only if the global map induced by $f$ on $G$ satisfies $\phi$.
\end{itemize}
Moreover, the pair $(\phi,f)$ can be effectively constructed from $\Psi$, and conversely.

Said differently, (FO,CA) pairs and MSO formulas define exactly the same graph languages, and the corresponding model checking problems are many-one equivalent on any fixed graph. We believe that this new characterization of MSO is particularly relevant in the context of symbolic dynamics. 

First, FO properties of orbits of CA are a conjugacy invariant and were much studied as mentioned above. We obtain a characterization of the decidability of FO model checking for CA orbits on Cayley graphs (Corollary~\ref{coro:focamodelchecking}) and we show that undecidability can be obtained with a fixed FO formula (Corollary~\ref{coro:onefo}) exactly on non virtually free f.g. groups, which should be put in perspective with the Ballier-Stein conjecture \cite{Ballier_2018}. We also obtain undecidability of a satisfiability problem for CA on finite graphs for a fixed FO formula (Corollary~\ref{coro:finitesatfoca}).

Besides, when fixing an arbitrary FO formula and letting the CA vary, we get fragments of MSO that make sense beyond the examples directly motivated by CA theory.
In particular, we show in Section~\ref{sec:cayley-graphs-domino} that on Cayley graphs of finitely generated groups, such fragments do not depend on the choice of generators, and that the domino problem and its variants (seeded and recurring) are equivalent to the model checking problem of some simple fixed FO formula (Theorem~\ref{theo:dominofoca}). For this we use an additional relation in the case of the recurring domino problem that remains MSO-expressible (Lemma~\ref{lem:msocayleyinfinite}). Finally we obtain a FO formula with just one quantifier alternation whose model checking problem does not belong to the arithmetical hierarchy when fixing the graph $\Z^2$ (Theorem~\ref{theo:foanadur}).

\subparagraph*{Warmup examples.} \label{sec:exa}
To fix ideas and give some intuition on how FO logic on CA orbits can be used to express graph properties, let us consider two well-known MSO properties and give an informal translation into a pair made of a FO formula and a CA local rule.

\begin{example}[$k$-Colorable graphs]
  Fix $k$ and consider an undirected graph ${G=(V,E)}$. Consider the CA local rule with state set ${S=\{0,\ldots, k-1\}}$ such that a vertex in state $i$ changes its state to ${i+1\bmod k}$ if it has a neighbor in state $i$ and remains in state $i$ otherwise. It can be checked that the CA induced on $G$ by this local rule has a fixed-point if and only if $G$ admits a proper vertex coloring with $k$ colors (\textit{i.e.} a coloring where no two neighboring vertices have the same color).
\end{example}

\begin{example}[Connected graphs]\label{exa:connected}
  Consider an undirected graph ${G=(V,E)}$. Consider the state set ${S=\{0,1,a_0,a_1,a_2\}}$ and the CA local rule such that a vertex in state ${i\in\{0,1\}}$ becomes ${1-i}$, a vertex in state $a_i$ becomes $0$ if it has a neighbor in ${\{0,1\}}$ and ${a_{i+1\bmod 3}}$ otherwise.
  We claim that $G$ is connected if and only if the CA induced on $G$ by this local rule has no periodic orbit of minimal period $6$, which is obviously an FO property of orbits.
\end{example}
\begin{proof}[Proof of the claim]
   If $G$ is not connected then we can define a configuration $c$ equal to $a_0$ on one component and $0$ on the others: $c$ is periodic of minimal period $6$ because no vertex in $a_i$ will ever have a neighbor in ${\{0,1\}}$, so these vertices will have a cyclic behavior of period $3$ and those from others components, a cyclic behavior of period $2$.
  If $G$ is connected, then no periodic orbit has minimal period $6$ because any configuration $c$ in such an orbit would necessarily contain a vertex in some state $a_i$ and another in a state from ${\{0,1\}}$ (otherwise the minimal period would be $2$ or $3$) and, since $G$ is connected, it would actually contain a vertex $v$ in state $a_i$ with a neighbor in a state from ${\{0,1\}}$. Then after one step this vertex $v$ would become $0$ and would never turn back to state $a_i$ by definition of $F$: this contradicts the fact that $c$ belongs to a periodic orbit.
\end{proof}

\subparagraph*{Comparison with another characterization of MSO by automata.} Several previous works established equivalence results between logic formalism and automata theory in the context of MSO languages of graphs. As mentioned above, \cite{Muller_1985} already noticed that some cellular automata properties can be translated into MSO. Following this, \cite{Thomas_1991} and \cite{Schwentick_1998} introduced tiling and automata recognizers that are equivalent to (small) fragments of MSO. In \cite{Reiter_2015}, alternating distributed graph automata are introduced that recognize exactly the languages of graphs definable in MSO logic. These distributed automata are close to CA in the sense that they run on configurations (coloring of the input graph by states) and use finite local memory and local communication between neighboring vertices. However, they are highly non-deterministic (alternating) and their accepting mechanism uses both initialization of the run to a particular configuration, and a global knowledge of the final configuration reached (precisely the set of states present in this configuration). Restrictions of this model to deterministic or non-deterministic automata (instead of alternating) gives strictly weaker fragments.

Our goal here is not to define a single automata model equivalent to MSO. Instead our approach motivated by symbolic dynamics uses deterministic CA on one hand and quantifier alternations in a separated FO formula which plays also the role of the accepting condition on the other hand. A key aspect is that we thus get natural fragments of MSO by fixing the FO formula and letting the CA rule vary. This also makes a strong difference with distributed alternating automata in the accepting mechanism since the FO formula does not offer any direct means of initializing some computation on a particular configuration, nor to detect presence of some particular states in a final configuration. 

\section{Formal definitions}
\label{sec:defs}

\subparagraph*{Graphs.} 
\newcommand\congraphs{\mathcal{C}}
 A ${(\Sigma,\Delta)}$-labeled graph is a graph ${G=(V,(E_\delta)_{\delta\in\Delta},L)}$, which can be finite or infinite, where ${L:V\to\Sigma}$ is the vertex labeling and ${E_\delta\subseteq V\times V}$ are the edges labeled by $\delta$. 
In such a graph, given some finite word ${w=w_1\cdots w_k\in\Delta^ \ast}$, a path labeled by $w$ is a sequence of vertices ${v_1,\ldots, v_{k+1}}$ such that, for any ${1\leq i\leq k}$, ${(v_i,v_{i+1})\in E_{w_i}}$.
All graphs considered in this paper are simple, meaning that there is at most one edge of a given label between two given vertices\footnote{The hypothesis that our graphs are simple will be used in sections~\ref{sec:coros} and \ref{sec:cayley-graphs-domino}. We choose to adopt this hypothesis across the entire paper for simplicity and clarity, however the main results from Section~\ref{sec:translation} should also hold without this hypothesis.}. Such a graph is said to be \emph{connected} if it is connected as an undirected and unlabeled graph, \textit{i.e.} if ${G=(V,E)}$ is connected where ${(v,v')\in E}$ if either ${(v,v')\in E_\delta}$ or ${(v',v)\in E_\delta}$ for some ${\delta\in\Delta}$. The set of connected graphs is denoted by ${\congraphs}$. $\Sigma$ will in some case be a singleton and can therefore be silently omitted: we speak about $\Delta$-labeled graph in this case. 
An important class of graphs studied in symbolic dynamics is that of Cayley graphs of finitely generated (f.g.) groups. Given a f.g. group $(\Gamma,\cdot)$ and a (finite) set of generators $\Delta$ (including their inverses), the associated Cayley graph is the $\Delta$-labeled graph where ${(\gamma,\gamma')\in E_\delta}$ if and only if ${\gamma'=\gamma\cdot\delta}$. 
An undirected graph $G_1$ is a minor of another undirected graph $G_2$ if it can be obtained from $G_2$ by deleting edges and vertices, and by contracting edges (\textit{i.e.} identifying the vertices incident to the edge without creating multiple edges).

\subparagraph*{Cellular automata.}
Given a finite set of states $S$ and a set of vertices $V$, a configuration is an element of $S^V$. It can be seen as a coloring of vertices by $S$.
A CA is a map from configurations to configurations, that is induced by a uniform local rule.
It is generally studied as a dynamical system through its set of orbits, which are sequences of configurations obtained by iterating the map from an initial configuration.

CA are usually defined over a fixed Cayley graph of a (f.g.) group \cite{cagroups}.
Following this classical approach, the local rule defining a CA is formally a lookup table and is bound to a particular graph as it relies on local patterns defined over bounded balls of the graph. It also requires the graph to be uniform and of bounded degree.
More general definitions were proposed that don't stick to a particular graph \cite{ARRIGHI_2017,Arrighi_2013}, but they still rely on the hypothesis of bounded degree and use a particular labeling by port numbers.

We introduce in this section a simple definition of CA on arbitrary labeled graphs. The key advantage of our formalism is that a given local rule actually defines a CA on \emph{any} labeled graph for fixed label sets. We can therefore fix a local rule and asks on which graphs the corresponding CA has a given property (as sketched in Section~\ref{sec:exa}): a pair made of a local rule and a property of CA dynamics actually defines a property of graphs. Moreover, on Cayley graphs of f.g. groups our formalism is equivalent to the classical one.
The CA local rules we consider can intuitively be seen as finite memory and finite distance exploring machines working as follows in parallel from each vertex $v$: they walk from $v$ following all possible $\Delta$-labeled walks up to some length $r$, and harvest states seen at the end of these walks and count their occurrences up to some constant $k$, then they decide from this information (a multiset) the new state at vertex $v$.
From this point of view, the edge labeling by $\Delta$ acts as local directions that give more information on the position in the graph given a labeled walk, while vertex labeling $\sigma$ gives some level of non-uniformity as in non-uniform cellular automata \cite{Dennunzio_2012}.


Let us now formalize this definition.
\newcommand\vreach{\mathcal{R}}
Given ${v\in V}$ and a word ${w\in\Delta^\ast}$, we denote by ${\vreach^w(v)}$ the set of vertices reachable from $v$ by a path labeled by $w$. If $w=\epsilon$ (empty word), then ${\vreach^w(v)=\{v\}}$ by definition. If $G$ is the Cayley graph of a f.g. group, then ${\vreach^w(v)}$ is always a singleton, however there are generally several paths reaching the same vertex.

\newcommand\mulset{\textsc{ms}}
\newcommand\patt{P}
\newcommand\mscap{\textsf{cap}}

The $k$-capped multisets over set $X$ are the multisets where no cardinality is greater than $k$, and are denoted ${\mulset^k(X) = {\{0,\ldots,k\}}^X}$.  Given a multiset ${m\in \N^X}$, we denote by ${\mscap^k(m)}$ the $k$-capped multiset such that ${\mscap^k(m)(x) = \max(m(x),k)}$. We denote by ${A^{\leq r}}$ the words of length at most $r$ over alphabet ${A}$ including the empty word $\epsilon$.  Given a set of states $S$ and a configuration ${c\in S^V}$ and ${v\in V}$, the $k$-capped pattern of radius $r$ at $v$ in $c$ is the $k$-capped multiset ${\patt(c,v,r,k)\in\mulset^k\bigl(\Delta^{\leq r}\times S\bigr)}$ defined by:
  \[\patt(c,v,r,k) = \mscap^k\bigl((w,s)\mapsto \#\{v'\in\vreach^w(v) : c_{v'}=s\}\bigr).\]

\begin{definition}[CA local rules and global maps]\label{def:CA}
  A CA local rule for ${(\Sigma,\Delta)}$-labeled graphs of state set $S$, radius $r$ and using $k$-capped multisets (${k\geq 1}$) is a map 
\[f : \Sigma\times\mulset^k\bigl(\Delta^{\leq r}\times S\bigr)\to S.\]
 For any ${(\Sigma,\Delta)}$-labeled graph with vertices $V$, the global CA map ${F_{G,f}: S^V\to S^V}$ associated to the local map $f$ and graph $G$ is then defined by 
\[F_{G,f}(c)_v = f(\sigma(c_v),\patt(c,v,r,k))\]
for any configuration ${c\in S^V}$ and any vertex $v\in V$.
\end{definition}

In the sequel, when considering a local map $f$, it always implicitly comes with a specified state set $S$ and values of $r$ and $k$ defining its domain and image sets.
We are mainly interested in the case where $\Sigma$ is a singleton (uniform CA), but incorporating $\Sigma$ in our definition allows non-uniformity in the local rule as in non-uniform CA \cite{Dennunzio_2012}.

\begin{remark}
On a Cayley graph of f.g. group, ${\patt(x,v,r,k)}$ gives all the information about configuration $x$ restricted to the ball in the graph centered in $v$ and with radius $r$: indeed, in this case, ${\patt(x,v,r,k)(q,w)=1}$ if and only if the unique vertex $v'\in\vreach^w(v)$ is such that ${x_{v'}=q}$, and $\vreach^w(v)$ describes the entire ball when ${w}$ enumerates $\Delta^{\leq r}$.
From this observation it follows that in the case of Cayley graphs and when $\Sigma$ is a singleton, the global CA maps from Definition~\ref{def:CA} are exactly the classical global CA maps (see for instance \cite{cagroups}).
\end{remark}

Definition~\ref{def:CA} explains how a given CA local rule $f$ induces a CA global map $F_{G,f}$ on a given graph. $F_{G,f}$ is the main object of study in CA theory, as it represents a dynamical system.

\begin{example}[Two definitions of Game of Life]
  Consider the famous Game of Life CA ${F:\{0,1\}^{\Z^2}\to\{0,1\}^{\Z^2}}$ \cite{life,Gardner_1970}. In this example $\Sigma$ is a singleton and thus ignored to simplify notations.
  First, let $G_1$ be the Cayley graph of $\Z^2$ with generators ${n=(0,1)}$, ${e=(1,0)}$ and their inverses. Let $M$ be the following set of words in $\Delta^{\leq 2}$: $n$, $n^{-1}$, $e$, $e^{-1}$, $ne$, $ne^{-1}$, $n^{-1}e$, $n^{-1}e^{-1}$. Now define the local rule $f_1$ of radius $2$ and using $4$-capped multiset by
  \[f_1(\mu) =
    \begin{cases}
      1 & \text{ if $\mu(1,\epsilon)=0$ and }\sum_{w\in M}\mu(1,w) = 3\\
      1 & \text{ if $\mu(1,\epsilon)=1$ and } 2\leq \sum_{w\in M}\mu(1,w) \leq 3\\
      0 & \text{ otherwise.}
    \end{cases}
  \]
One can check that ${F_{G_1,f_1}= F}$.
Now consider the undirected and unlabeled graph $G_2=(\Z^2,E)$ with ${((i,j),(i',j'))\in E}$ if ${|i-i'|\leq 1}$ and ${|j-j'|\leq 1}$ and ${(i,j)\neq (i',j')}$. Here we denote $\Delta=\{u\}$. We then define another local rule $f_2$ of radius $1$ and using $4$-capped multiset as follows:
\[f_2(\mu) = 
    \begin{cases}
      1 & \text{ if $\mu(1,\epsilon)=0$ and } \mu(1,u) = 3\\
      1 & \text{ if $\mu(1,\epsilon)=1$ and } 2\leq \mu(1,u) \leq 3\\
      0 & \text{ otherwise.}
    \end{cases}
  \]
One can again check that ${F_{G_2,f_2} = F}$.
\end{example}

\subparagraph*{Logics.}
\newcommand\relsp[1]{\,#1\,}

To make the exposition more concise, we suppose some familiarity with standard concepts of formal logic (variables, assignments, quantification, free variables, etc).
MSO logic uses first-order variables (usually denoted by lower-case letters) representing vertices and second-order variables (usually denoted by upper-case letters) representing sets of vertices.
To help reading, relations in formulas will use infix notation (${x\relsp R y}$) while relation in the meta-language will use the set notation (${(x,y)\in R}$).

\begin{definition}[MSO formulas and their semantics]
  The set MSO formulas over label sets $(\Sigma,\Delta)$ is the set of atomic formulas:
  \begin{itemize}
  \item ${x\relsp{L} \sigma}$ for $x$ a first-order variable and ${\sigma\in\Sigma}$ (meaning $x$ has label $\sigma$),
  \item ${x \relsp{E_\delta} x'}$ for $x$ and $x'$ first-order variables and ${\delta\in\Delta}$ (meaning that there is an edge labeled $\delta$ from $x$ to $x'$),
  \item ${x=x'}$ for first-order variables $x$ and $x'$ (meaning that $x$ is equal to $x'$),
  \item ${x\in X}$ for first-order variable $x$ and second-order variable $X$ (meaning that $x$ belongs to set $X$),
  \end{itemize}
  closed by the usual logic connectives ($\vee$, $\wedge$, $\neg$) and quantifiers ($\forall$, $\exists$).
  Given an MSO formula $\Psi$, a ${(\Sigma,\Delta)}$-labeled graph $G$ and an assignment $\alpha$ of free variables of $\Psi$ we define the semantics in the standard way starting from the obvious meaning of atomic formula above (see \cite{courcellebook} for an in-depth introduction). We write ${(G,\alpha)\models \Psi}$ when $\Psi$ is true on $G$ with assignment $\alpha$. If $\Psi$ has no free variable, we simply write ${G\models \Psi}$ when $\Psi$ is true on $G$.
\end{definition}

We will sometimes use substitution of relations with formulas defining them. For instance we can write ${\Psi(X,\Psi_R(x_1,x_2))}$, where $\Psi$ is a formula using an additional relation symbol $R$, to denote the MSO formula obtained by substituting $\Psi_R$ for $R$ in $\Psi$ (with the usual precaution of renaming variables if necessary, see \cite{courcellebook}).

We now define FO logic over orbits of CA: they are just formulas allowing quantifications over configurations and using two relations, equality and application of one step of the CA global rule.

\begin{definition}[FO formulas and their semantics]
  The set of FO formulas is made of atomic formulas:
  \begin{itemize}
  \item ${y = y'}$ (meaning that configuration $y$ is equal to configuration $y'$),
  \item ${y\to y'}$ (meaning that the global CA map leads to $y'$ in one step starting from $y$),
  \end{itemize}
  and closed under the usual logic connectives and quantifiers. Given a FO formula $\phi$, a CA global map ${F:S^V\to S^V}$ and an assignment ${\beta}$ of free variables of $\phi$ to configurations from ${S^V}$, we write ${(F,\beta)\models \phi}$ to denote that $F$ satisfies $\phi$ with assignment $\beta$ following the obvious semantics of formulas starting from the relations above. When $\phi$ has no free variable, we just write ${F\models\phi}$.
\end{definition}

For a CA global map ${F:S^V\to S^V}$ and a FO formula ${\phi}$ with free variables ${(y^1,\ldots,y^n)}$, we use the shortcut ${F\models \phi(c^1,\ldots,c^n)}$ for configurations $c^1,\ldots, c^n\in S^V$ to express that ${(F,\beta)\models \phi}$ where $\beta$ is the assignment given by ${y^i\mapsto c^i}$ for ${1\leq i\leq n}$. We will also use the FO shortcut
\newcommand\stepsdiff[1]{\to^{#1}_{\neq}}
$y^0 \stepsdiff{k} y^k$ to represent the FO formula expressing that $y^0$ leads to $y^k$ in $k$ steps and the $k+1$ configurations involved in this partial orbit are pairwise different:
\[y^0 \stepsdiff{k} y^k\isbydef \exists y^1,\ldots, \exists y^k : \bigwedge_{0\leq i<k}y^i\to y^{i+1} \wedge \bigwedge_{i\neq j}y^i\neq y^j.\]

\emph{Notation convention:} we will always use letter $\Psi$ for MSO formulas, $x$ or $X$ for MSO variables, $\alpha$ for MSO assignments, $\phi$ for FO formulas, $y$ for FO variables, $G$ for graphs, $f$ for CA local rules, $F$ for CA global maps and $c$ for configurations. We use notation $c_v$ to denote state of configuration at vertex $v$, that's why we prefer the exponent notation $c^1, c^2,\ldots$ to denote several configurations.

\subparagraph*{Combining graphs, CA and logics.}

The above definitions suggest various definitions of sets of objects (or languages):
 the graph languages 
 \begin{itemize}
 \item ${\graphs{\Psi} = \{G : G\models \Psi\}}$ and
 \item ${\graphs{\phi,f} = \{G : F_{G,f}\models \phi\}}$,
 \end{itemize}
and the set of CA local rules ${\cas{\phi,G} = \{f : F_{G,f}\models \phi\}}$,
where we use the notation convention above, and where $\Sigma$ and $\Delta$ are fixed so graphs are actually ${(\Sigma,\Delta)}$-graphs and CA local rules are rules for such graphs.
Moreover, $\cas{\phi,G}$ can be seen as decision problems where inputs are given as local maps of CA (model checking problem of $\phi$ on $G$).

\newcommand\domino{\mathcal{D}} 
\newcommand\turingequiv{\equiv_T}

\section{Translation results}
\label{sec:translation}

\subsection{From FO/CA pairs to MSO}

Whatever the state set, a CA configuration can be represented as a tuple of vertex sets: we can code the state at a vertex by the number of sets it belongs to among the tuple. This way, FO variables can easily be translated into tuples of second-order MSO variables undergoing the same quantification and  we get an onto map from possible assignments of the tuple of second-order MSO variables onto possible assignments of the corresponding FO variable.

Under that coding, equality of configurations translates into a simple MSO formula with just one universal first-order quantifier. It remains to show that the other relation in the signature of FO, relation $\to$ which represents the application of one step of the CA global map, can also be translated into MSO: 
this boils down to checking that at each vertex the local rule is correctly applied, which itself boils down to counting up to some constant occurrences of states that can be reached by a labeled walk of bounded length.

\begin{theorem}\label{theo:FOCAtoMSO}
  There is a recursive translation $\tau$ from pairs ($\phi$,$f$) made of a FO formula $\phi$ and a CA local rule $f$ to MSO formulas such that the following equivalence holds for any graph $G$: ${F_{f,G}\models\phi\iff G\models \tau(\phi,f)}$.
\end{theorem}
\newcommand\interpret{\iota}

\begin{proof}
  Fix $\phi$ and $f$ with state set $S=\{1,\ldots,n\}$, radius $r$ and using $k$-capped multisets. Suppose that $\phi$ is in prenex normal form where ${\phi^0(y^1,\ldots,y^m)}$ is the quantifier-free part (if not just compute this prenex normal form).
We will use the following map ${\interpret}$  from $n$-tuple of vertex sets ${\subseteq V}$ to configurations of ${S^V}$:
  \[\interpret(E_1,\ldots, E_n) = v\mapsto \#\{i : 1\leq i\leq n \wedge v\in E_i\}.\]
  $\iota$ is an onto map. Following this coding we define the MSO formula $\Psi_{s-\mathtt{state}}(x,\overline{X})$ expressing that the configuration represented by $\overline{X}$ at vertex $x$ is in state $s$:
  \[\Psi_{s-\mathtt{state}}(x,\overline{X})\isbydef \bigvee_{
        \overset{A\subseteq S}{|A|=s}} \bigwedge_{1\leq i\leq n} x\in X_i\iff i\in A.\]
  
For each relation ${y^1 \relsp{R}y^2}$ in the FO signature, we want to define an MSO formula 
$\Psi_R(\overline{X^1},\overline{X^2})$ where $\overline{X^i}$ are $n$-tuples, and with the following property: for any assignment $\alpha$ of this two $n$-tuples, and for any graph $G$, the assignment $\beta$ which assigns ${\interpret(\alpha(\overline{X^i}))}$ to $y^i$ verifies:
\[(G,\alpha)\models\Psi_R(\overline{X^1},\overline{X^2})\iff (F_{G,f},\beta)\models y^1\relsp{R} y^2.\]

First it is not difficult to write such a formula in the case of relation ${y^1 \relsp{=} y^2}$: \[\Psi_=(\overline{X^1},\overline{X^2}) \isbydef \forall x, \bigvee_{s\in S}\Psi_{s-\mathtt{state}}(x,\overline{X^1})\wedge \Psi_{s-\mathtt{state}}(x,\overline{X^2})\]
For the case of relation ${y^1\relsp{\to} y^2}$, we first need for each ${w\in\Delta^l}$ an MSO formula $\Psi_{w}(x_0,x_{l})$ expressing that vertex $x_l$ can be reached by a walk labeled $w$ starting from label $x_0$, that can be written as:
\[\Psi_{w}(x_0,x_l)\isbydef \exists x_1,\ldots, x_{l-1} : \bigwedge_{0\leq i<l}x_i\relsp{E_{w_i}} x_{i+1}.\]
Then we define the formula ${\Psi_{\patt(s,w)\geq p}(\overline{X},x)}$ for each ${0\leq p<k}$ by
\[\Psi_{\patt(s,w)\geq p}(\overline{X},x)\isbydef \exists x_1,\ldots, x_p : \bigwedge_{\overset{1\leq i,j\leq p}{i\neq j}}x_i\neq x_j \wedge \bigwedge_{1\leq i\leq p}\Psi_w(x,x_i) \wedge \Psi_{s-\mathtt{state}}(x_i,\overline{X})\]
which expresses that in configuration $\overline{X}$, there are at least $p$ occurrences of state $s$ reachable by a walk labeled $w$ starting from $x$.
We can write by Boolean combination of the above a formula $\Psi_{\patt(s,w)=p}$ expressing that in configuration $\overline{X}$, there are exactly $p$ occurrences of state $s$ reachable by a walk labeled $w$ starting from $x$.

We can now define the formula $\Psi_{\to}$ with the property announced above with respect to the FO relation $y^1\relsp{\to} y^2$ :
\[\Psi_{\to}(\overline{X^1},\overline{X^2})\isbydef \forall x,\bigwedge_{\overset{\sigma\in\Sigma}{\mu\in M}}\bigl(x\relsp{L}\sigma \wedge \bigwedge_{\overset{s\in S}{w\in\Delta^{\leq r}}}\Psi_{\patt(s,w)=\mu(s,w)}(\overline{X},x)\bigr)\Rightarrow \Psi_{f(\sigma,\mu)-\mathtt{state}}(x,\overline{X^2})\]
where ${M=\mulset^k\bigl(\Delta^{\leq r}\times S\bigr)}$. This formula is just the translation of the definition of a CA global rule from its local rule (Definition~\ref{def:CA}) saying that at each vertex, if the local symbol and multiset pattern are ${(\sigma,\mu)}$, then the next state should be ${f(\sigma,\mu)}$.

With these ingredients, $\phi^0(y^1,\ldots, y^m)$, the quantifier free part of $\phi$, is translated into an MSO formula $\Psi^0(\overline{X^1},\ldots,\overline{X^m})$ by replacing each occurrence of some atomic formula $y^i\relsp{=} y^j$ (resp. $y^i\relsp{\to} y^j$) by $\Psi_{=}(\overline{X^i},\overline{X^j})$ (resp. $\Psi_{\to}(\overline{X^i},\overline{X^j})$). We still have that, for any MSO assignment $\alpha$, the FO assignment $\beta$ defined with a slight abuse of notations as $\iota\circ\alpha$ as above verifies:
\[(G,\alpha)\models\Psi^0(\overline{X^1},\ldots,\overline{X^m})\iff (F_{G,f},\beta)\models \phi^0(y^1,\ldots, y^m).\]

Since $\iota$ is an onto map, one has also that
\[(G,\alpha')\models\forall\overline{X^1},\Psi^0(\overline{X^1},\ldots,\overline{X^m})\iff (F_{G,f},\beta')\models \forall y^1,\phi^0(y^1,\ldots, y^m)\]
  for any MSO assignment $\alpha'$ and corresponding FO assignment $\beta'$ defined as $\iota\circ\alpha'$ (with a slight abuse of notations).
The same hold with the $\exists$ quantifier and the Theorem follows by a straightforward induction.
\end{proof}

\subsection{From MSO to FO/CA pairs.}

This converse translation is less straightforward. Let us first give a simplified overview by considering an MSO formula $\Psi$ in prenex normal form, and describing its translation into a CA local rule $f$ together with a FO formula $\phi$.

We will use binary configurations (\textit{i.e.} elements of ${\{0,1\}^V}$) to code either second-order variable assignments (a set coded by its indicator function) or first-order variable assignments (a singleton coded by its indicator function).
More generally, we can code several variable assignments in a configuration made of several binary components, \textit{i.e.} configurations over a product alphabet ${S=\{0,1\}\times\cdots \times\{0,1\}}$.
Given a configuration made of a product of binary components coding an assignment of several variables, the truth of an atomic formula using these variables over this assignment can be checked by a CA local rule in a distributed manner. With slightly more work and using a particular FO property of orbits, we can actually test any quantifier free formula in a distributed manner.

The CA local rule together with the FO formula we are going to construct will essentially enforce an erasing process that starts from a configuration made of a product of binary components (to code an assignment of all MSO variables at once) and then removes components of the product one by one at successive steps until reaching a fixed point: on one hand having all information about variables assignment at the start allows to check the truth of the quantifier-free matrix of the MSO formula as hinted before, and, on the other hand, having components to disappear individually in successive steps allows to make a FO quantification over configurations following exactly the MSO quantification over variables. 

For instance, taking MSO formula ${\Psi=\forall X_1,\exists x_2, \forall X_3, R(X_1,x_2,X_3)}$, we construct a FO formula that is essentially of the form:
\[\forall y_1, \exists y_2, \forall y_3 : y_3\to y_2\to \begin{tikzpicture}[anchor=base, baseline]
    \node (a) {$y_1$};
    \draw (a) edge[loop right] (a);
  \end{tikzpicture}
\]
and a CA local rule that will ensure that $y_1\approx a_1$, $y_2\approx (a_1,a_2)$ and ${y_3\approx (a_1,a_2,a_3)}$ where $a_1$ is an assignment for $X_1$, $a_2$ is an assignment for $x_2$ and $a_3$ is an assignment for $X_3$.
It will also ensure, when in configuration $y_3$, that the assignments ${(a_1,a_2,a_3)}$ satisfy ${R(X_1,x_2,X_3)}$.
It is important to note that successive choices of assignments of variables $y_1$, $y_2$ and $y_3$ corresponds, up to a simple product encoding, to successive choices of assignments for variables $X_1$, $x_2$ and $X_3$.  
Non-deterministic choices of successive assignments are possible in this construction because they correspond to going backward in time in the canonical orbit enforced by the FO formula above: there is no contradiction with the determinism of cellular automata.

To turn this overview into a concrete construction, several technical points have to be addressed:
\begin{itemize}
\item the simplified behavior described above only works on some well-formed configuration; as usual in CA constructions, we will use local error detection and special error states to mark orbits of bad configurations and distinguish them from good ones: here we use two error states that oscillate with period two in order to ensure that any orbit reaching a fixed point has successfully passed all error detection mechanisms.
\item to code first-order variables, binary configurations need to have exactly one vertex in state 1 and the local nature of CA prevents from verifying this (it cannot a priori distinguish a configuration with a single 1 from a configuration with two 1s arbitrarily far away, not to mention the case of non-connected graphs). Our construction handles this through the FO formula to be satisfied using a sibling configurations counting trick combined with a particular behavior of the CA which uses additional layers of states.

  The erasing process of the CA mentioned above will therefore take several steps for first-order variables, and only one step for second-order variables.
\item checking a quantifier-free formula given an assignment of its variables encoded in a configuration is generally not doable in one step by a CA, especially on non-connected graphs that prevent the CA from communicating between components (think of the example: ${x\in X\vee y\in X}$) ; to solve this problem our construction will once again rely on a combination of FO logic over several steps and a particular behavior of the CA. 
\end{itemize}
We first give a solution to these technical problems that works when we restrict to connected graphs. This construction is a little simpler than the general case that we address later, and it has the benefit to induce a better controlled  dependence of the FO formula on the MSO formula (an aspect that will turn out to be useful in Section~\ref{sec:coros}).

If $\Psi$ is an MSO formula in prenex normal form with quantifier prefix ${Q_1,\ldots, Q_n}$, its prefix signature is the word describing the alternance of quantifiers taking into account both the type of quantification and the order of quantified variables. More precisely, it is the word over alphabet ${\{\forall,\exists\}\times\{1,2\}}$ obtained as follows: first map each quantifier to the alphabet according to the actual type and order, then remove any repetition of consecutive identical letters.

\begin{theorem}\label{theo:MSOFOCACON}
  There are two recursive transformations $\tau_{FO}$ from MSO formulas to FO formulas and $\tau_{CA}$ from MSO formulas to CA local rules such that, for any MSO formula $\Psi$, the pair made of ${\phi=\tau_{FO}(\Psi)}$ and ${f=\tau_{CA}(\Psi)}$ verifies:
  \begin{enumerate}
  \item for any connected graph $G$ the following equivalence holds: ${G\models\Psi\iff F_{G,f}\models \phi}$,
  \item if $\Psi$ is prenex then $\phi$ depends only on its prefix signature.
  \end{enumerate}
\end{theorem}

Let 
${\Psi = \quantif_1x_1^1,\ldots, \quantif_1x_1^{k_1},\quantif_2x_2^1,\ldots, \quantif_2x_2^{k_2},\ldots, \quantif_nx_n^1,\ldots, \quantif_nx_n^{k_n}, R(x_1^1,\ldots,x_n^{k_n})}$
be any MSO formula in prenex normal form where $\quantif_1,\ldots, \quantif_n$ are the $n$ quantifiers types (either $\forall$ or $\exists$ and either first or second-order) forming the prefix signature, variables $x_i^1,\ldots, x_i^{k_i}$ are bound by quantifier of type $\quantif_i$ and $R(x_1^1,\ldots,x_n^{k_n})$ is the matrix of the prenex normal form (i.e. a quantifier free formula). We use this numbering of variables grouped by quantifiers type to obtain a more compact FO formula that only depends on the prefix signature of $\Psi$.
Let's write $R(x_1^1,\ldots,x_n^{k_n})$ in disjunctive normal form:
\[R(x_1^1,\ldots,x_n^{k_n}) = \bigvee_{1\leq j\leq d} C_j(x_1^1,\ldots,x_n^{k_n})\]
where each clause $C_j$ is a conjunction of terms which are atomic formula or negation thereof using variables ${x_1^1,\ldots,x_n^{k_n}}$.
\newcommand\foind{\mathcal{O}}
\newcommand\llayers{\lambda}
Let ${\foind}$ be the set of $i$ such that $\quantif_i$ is a first-order quantifier.

\subparagraph*{Structure of configurations.}

For each ${1\leq i\leq n}$, let ${\omega(i)=\bigl|\foind\cap \{1,\ldots,i\}\bigr|}$ and define ${\llayers(i) = i + 2\cdot \omega(i)}$. As it will become clear below, $\llayers(n)$ denotes the length of an orbit along which $n$ particular configurations will be identified. Configurations along this orbit will use distinct state sets, and $\llayers(i)$ will also denotes the number of layers of the $i$-th configuration. We first introduce sets $S_l$ for ${1\leq l\leq \llayers(n)}$ that will be used to hold variable assignments and translate MSO quantification over variables of first or second-order into FO quantification over configurations. $S_l$ is a product of $l$ \emph{layers} each of the form ${\{0,1\}^{k_i}}$ (\emph{variable} layer) or ${\{1,\ldots,k_i\}}$ (\emph{choice} layer) for some $i$, or $\{0,1\}$ (\emph{control} layer). Intuitively, variable layers will hold MSO variables assignments, and choice and control layers are used only for first-order variables as a control mechanism. Sets $S_l$ are precisely defined as follows:
\begin{itemize}
\item $S_1 = \{0,1\}^{k_1}$ and if $1\in\foind$ then $S_2=S_1\times\{1,\ldots,k_1\}$ and $S_3=S_2\times\{0,1\}$,
\item for ${1\leq i<n}$, $S_{\llayers(i)+1}=S_{\llayers(i)}\times\{0,1\}^{k_{i+1}}$ and if ${i+1}\in\foind$ then $S_{\llayers(i)+2}=S_{\llayers(i)+1}\times\{1,\ldots,k_{i+1}\}$ and $S_{\llayers(i)+3}=S_{\llayers(i)+2}\times\{0,1\}$.
\end{itemize}

 If ${j=\llayers(i) \leq l}$ with $i\not\in\foind$ or ${j=\llayers(i)-2\leq l}$ with $i\in\foind$ then the $j$-th layer of $S_l$ is a \emph{variable} layer, denoted as ${V_i(S_l)}$, and intuitively represents an assignment for the tuple of variables ${x_i^1,\ldots, x_i^{k_i}}$. For ${1\leq j\leq k_i}$, we denote by ${V_i^j(S_l)}$ the $j$-th binary component of ${V_i(S_l)}$ which intuitively represents an assignment for variable $x_i^j$. For $i\in\foind$ and ${j=\llayers(i)-1\leq l}$, the $j$th layer of $S_l$ is a choice layer, denoted $\chi_i(S_l)$. Other layers, precisely $j$-th layers with ${j=\llayers(i)}$ with $i\in\foind$, are \emph{control} layers and denoted $K_i(S_l)$. Choice layer $\chi_i$ together with control layer $K_i$ are used to ensure that the corresponding variable layer $V_i$ correctly encodes a $k_i$-tuple of assignments of first-order variables, \textit{i.e.} is a $k_i$-tuple of binary configurations each having exactly one position in state $1$. 

 We denote by $\pi$ the natural projection from ${S_{l}}$ onto ${S_{l-1}}$ (for any ${2\leq l\leq \llayers(n)}$) which removes the last ($l$-th) layer of elements of $S_l$.

A well-formed configuration where all vertices are in a state from $S_{\llayers(n)}$ intuitively contains an assignment for all variables involved in formula $\Psi$ (provided the control mechanism for first-order variables to be detailed below has been successful). We need to implement a distributed check of the truth  of quantifier-free formula $R$ on such an assignment. The key is to ensure that some clause $C_j$ from the disjunctive normal form of $R$ is chosen uniformly on the entire graph and to check everywhere that each terms of $C_j$ is locally correct given the assignment. 
For ${1\leq j\leq d}$, let ${T_j = S_{\llayers(n)}\times\{j\}}$ (recall that $d$ is the number of clauses in the disjunctive normal form of $R(x_1^1,\ldots, x_n^{k_n})$). We again use notation $\pi$ to denote the natural projection from ${T_j}$ onto $S_{\llayers(n)}$, which removes the last component of states. ${T_j}$ sates will be used to check clause $C_j$.

We can now define the state set of the CA ${\tau_{CA}(\Psi)}$ as
\[S = \{e_0,e_1\}\cup \bigcup_{1\leq l\leq \llayers(n)} S_l \cup \bigcup_{1\leq j\leq d} T_j\]
where $e_1$ and $e_2$ are distinct elements from the rest of $S$. We say that the \emph{type} of an element of ${S}$ is $l$ if it belongs to $S_l$, \emph{error} if it is $e_0$ or $e_1$ and \emph{$j$-truth-check} if it belongs to ${T_j}$. We also naturally extend the notation $V_i$, $\chi_i$  and $K_i$ for any state $s\in T_j$ by ${V_i(s) = V_i(\pi(s))}$,  ${\chi_i(s)=\chi_i(\pi(s))}$ and ${K_i(s)=K_i(\pi(s))}$. 

A configuration ${c\in S^V}$ is \emph{valid} if the following conditions hold:
\begin{itemize}
\item states of all pairs of neighboring vertices of $G$ are of the same type, and not of error type;
\item choice layer $\chi_i$ of all pairs of neighboring vertices of $G$ are equal;
\item at each vertex $v$ the control layers have zeros where the corresponding variable layers indicated by the choice layers have, precisely: ${K_i(c_v)\leq V_i^{\chi_i(c_v)}(c_v)}$ for all $i$ such that ${K_i(c_v)}$ is defined (intuitively, a $1$ in a control layer is authorized only if their is a $1$ in the 'chosen' component of the corresponding variable layer);
\end{itemize}
\newcommand\locvalid{\texttt{valid}}
Note that this definition of validity is purely local. For $\mu$ a capped multiset (second argument of the local rule of a CA), we write ${\locvalid(\mu)}$ to express the local validity conditions above on $\mu$: the first two item are checked on each pair of states ${(s,s')}$ such that ${\mu(s,\epsilon)\geq 1}$ and ${\mu(s',\delta)\geq 1}$ for some ${\delta\in\Delta}$; the third item is checked on state $s$ such that ${\mu(s,\epsilon)\geq 1}$.

\subparagraph*{CA local rule.}
The behavior of the CA local rule ${f=\tau_{CA}(\Psi)}$ can intuitively be described as follows:
\begin{itemize}
\item check the local validity of the configuration and if not generate states of error type that alternate with period $2$ (between $e_0$ and $e_1$),
\item apply projection $\pi$ on states of type $l$ with ${l\geq 2}$ and let states of type $1$ unchanged,
\item on states of type ${j}$-truth-check, verify that the assignment of variables ${x_1^1,\ldots, x_n^{k_n}}$ coded in variable layers $V_i^j$ are such that ${C_j(x_1^1,\ldots, x_n^{k_n})}$ holds, and apply $\pi$ if it is the case, or generate an error state otherwise.
\end{itemize}

Of course, the behavior on states of type ${j}$-truth-check above has to be understood locally since we are defining a CA. The implementation of this distributed truth check is as follows: each term ${t(x_a^b,x_p^q)}$ appearing in clause $C_j$ (an atomic formula or its negation), where $x_a^b$ is a first-order variable, is checked only at any vertex having a $V_a^b$ component at $1$, otherwise it is considered true by default.
\newcommand\locallymodels{\models_{loc}}
 More precisely, for a pair ${(\sigma,\mu)}$ made of a vertex label and a capped multiset (arguments of the local rule $f$), we write ${(\sigma,\mu)\locallymodels C_j}$ (clause $C_j$ is locally valid) if the local state (unique $s$ such that ${\mu(s,\epsilon)}$)  is of type ${j}$-truth-check and the neighboring states also, and if all terms of $C_j$ are locally true according to the previous rule.
It turns out that ${(\sigma,\mu)\locallymodels C_j}$ can be checked by a local rule of CA of radius $1$ and using $1$-capped multisets (\textit{i.e.} sets). Precisely, all possible terms of clause $C_j$ are treated as follows (denoting again $s$ the unique state such that ${\mu(s,\epsilon)\geq 1}$):
\begin{itemize}
\item ${x_a^b \relsp L \sigma'}$ (resp. its negation) is true if and only if $V_a^b(s)=0$ or if $\sigma=\sigma'$ (resp. $\sigma\neq\sigma'$),
\item ${x_a^b \relsp{E_\delta} x_p^q}$ (resp. its negation) is true if and only if $V_a^b(s)=0$ or if ${1\in\{V_p^q(q') : \mu(q',\delta)\geq 1\}}$ (resp. $1$ does not belong to this set),
\item ${x_a^b = x_p^q}$ (resp. its negation) is true if and only if $V_a^b(s)=0$ or ${V_p^q(s)=1}$ (resp. ${V_p^q(s)=0}$),
\item ${x_a^b \in x_p^q}$ (resp. its negation) is true if and only if ${V_a^b(s)=0}$ or ${V_p^q(s)=1}$ (resp. ${V_p^q(s)=0}$).
\end{itemize}
For $t$ a term, we write ${(\sigma,\mu)\locallymodels t}$ if $t$ is locally true according to the above definition.
Recall that $x_a^b$ is a first order variable and the rest of the construction will ensure that, on configurations that matter, there will always exist a node at which ${V_a^b(s)=1}$ so these tests will actually check that the assignments of variables encoded in the configuration do satisfy the term as desired.

The CA local rule ${f}$ is then defined as follows (denoting again $s$ the unique state such that ${\mu(s,\epsilon)\geq 1}$):
\[f(\sigma,\mu) =
  \begin{cases}
    e_{1-i} &\text{ if $s=e_i$}\\
    e_0 &\text{ otherwise, and if $\neg\locvalid(\mu)$,}\\
    e_0 &\text{ otherwise, and if $s$ has type ${j}$-truth-check and ${(\sigma,\mu)\not\locallymodels C_j}$,}\\
    s &\text{ otherwise, and if $s$ is of type $1$,}\\
    \pi(s) &\text{ otherwise.}
  \end{cases}
\]




\subparagraph*{FO formula.}
\newcommand\sequ[2]{\texttt{seq}_{#1}(#2)}
\newcommand\sequple[2]{\texttt{good}_{#1}(#2)}
\newcommand\truth[1]{\texttt{truth}(#1)}
\newcommand\preimj[2]{\texttt{preimg}_{#1}(#2)}
\newcommand\siblings[1]{\texttt{\#siblings}(#1)}
\newcommand\okfov[1]{\texttt{goodFOVAR}(#1)}

Most of the task of the FO formula is to check that $n$ configurations are well-positioned in an orbit leading to a fixed-point. However, along this orbit, we also have to make checks to ensure that layers corresponding to first-order variables are well formed. For ${i\in\foind}$, the CA behavior already ensures (by generating error states if not) that choice layers $\chi_i$ are uniform and that control layers $K_i$ are upper-bounded by the chosen corresponding variable layer $V_i^j$. In this context, the check is done as follows (intuitively, variable $y^i$ represents a configuration of type $\llayers(i)$ at each vertex, for some $i\in\foind$):
\[\okfov{y^i}\isbydef \forall y, \forall y', (y^i\stepsdiff{2} y'\wedge y\stepsdiff{2} y')\Rightarrow\siblings{y}=1,\]

where the formula ${\siblings{y}=1}$ expresses that there is exactly $1$ configuration other than $y$ with same image as $y$ and can be written explicitly in FO as follows: \[\exists y_s,\exists y_+, y_s\to y_+\wedge y\to y_+ \wedge y_s\neq y \wedge (\forall y': y'\to y_+ \Rightarrow (y'=y \vee y'=y_s)).\] The idea is that a variable layer $V_i$ is good if, for any choice $j$ made in choice layer $\chi_i$, there are only $2$ possible ways to correctly complete the control layer $K_i$, because there is exactly one vertex $v$ at which ${V_i^j}$ is $1$ and therefore at which $K_i$ can be freely chosen to be $0$ or $1$.

Let us now define formulas to deal with the global structure of the orbit leading to a fixed point:
\[\sequ{1}{y} \isbydef
  \begin{cases}
    y\to y&\text{ if $1\not\in\foind$},\\
    \exists y^0, y\stepsdiff{2} y^0\wedge y^0\to y^0\wedge\okfov{y} &\text{ if $1\in\foind$,}
  \end{cases}
\]
and for any ${2\leq i\leq n}$ 
\[\sequ{i}{y,y^+}\isbydef
  \begin{cases}
    y^+\to y & \text{ if }i\not\in\foind\\
    y^+\stepsdiff{3} y \wedge \okfov{y^+} & \text{ if }i\in\foind\\
  \end{cases}.\]

Then, denote by ${\sequple{i}{y^1,\ldots, y^i}}$ for each ${1\leq i\leq n}$ the formula:
\[\sequple{i}{y^1,\ldots, y^i}\isbydef\sequ{1}{y^1}\wedge\bigwedge_{2\leq k\leq i}\sequ{k}{y^{k-1},y^k}.\]

The truth check for $R$ has to make a non-deterministic choice of clause $C_j$ to then use the distributed truth check implemented in the CA, we therefore define the following formula to be used in $\phi$: ${\truth{y}\isbydef \exists y'\to y}$
which make sense when $y$ represents a configuration everywhere of type ${\llayers(n)}$.

As we show later, formula ${\sequple{i}{y^1,\ldots, y^i}}$ paired with CA $F$ ensures that the configuration assigned to $y^i$ is a well-formed configuration of type $\llayers(i)$ that holds an assignment for variables ${(x_i^1,\ldots, x_i^{k_i})}$ through its variable components ${V_i^1,\ldots, V_i^{k_i}}$. We use these formulas in $\phi$ to make restricted domain FO quantifications that exactly correspond to well-formed configurations that hold assignments of the corresponding MSO variables. 
  
To make the formula $\phi$ more readable, we use the following syntactic sugar to express restricted domain quantification. If $\phi_D$ and $\phi$ are formulas containing $y$ as free variable, then:
\begin{itemize}
\item ${\exists y\in\phi_D, \phi}$ stands for ${\exists y, \phi_D\wedge\phi}$,
\item ${\forall y\in\phi_D, \phi}$ stands for ${\forall y, \phi_D\Rightarrow \phi}$.
\end{itemize}

We can finally define FO formula  $\phi=\tau_{FO}(\Psi)$:

\[\phi\isbydef \quantif_1' y^1\in\sequple{1}{y^1}, \quantif_2' y^2\in\sequple{2}{y^1,y^2},\ldots, \quantif_n' y^n\in\sequple{n}{y^1,\ldots,y^n}, \truth{y^n}\]

where the FO quantifier $\quantif_i'$ is $\exists$ if the MSO quantifier $\quantif_i$ is existential, and $\forall$ if $\quantif_i$ is universal.

\subparagraph*{Correctness of the construction.} First, it can be checked that $\phi$ only depends on the prefix of $\Psi$ and not on $R$. Second, the construction of $f$ and $\phi$ are clearly computable from $\Psi$. The proof of Theorem~\ref{theo:MSOFOCACON} then relies on two lemmas. The first one ensures that ${\sequple{i}{\ldots}}$ predicates correctly translate assignments of FO variables quantified in $\phi$ into assignments of MSO variables quantified in $\Psi$ and conversely. 

\begin{lemma}\label{lem:foquantifassign}
  Let $G$ be a ${(\Sigma,\Delta)}$-labeled graph which is connected. Consider configurations ${c^1,\ldots, c^i}$ with ${1\leq i\leq n}$ such that ${F_{G,f}\models \sequple{i}{c^1,\ldots, c^i}}$, then the following holds:
  \begin{enumerate}
  \item $c^k$ is of type $\llayers(k)$ at each vertex, for ${1\leq k\leq i}$;
  \item  for ${1\leq l\leq i}$ and ${1\leq j\leq k_l}$, variable component $V_l^j(c^k)$ is the same for all $k$ with ${l\leq k\leq i}$, and it is such that exactly one vertex is in state $1$ when $l\in\foind$ ;
  \item if ${i<n}$, for any assignment $\alpha$ of variables ${(x_{i+1}^1,\ldots,x_{i+1}^{k_{i+1}})}$ there exists $c^{i+1}$ such that \[F_{G,f}\models \sequple{i+1}{c^1,\ldots ,c^{i+1}}\] and ${V_{i+1}^j(c^{i+1}) = \alpha(x_{i+1}^j)}$ for all ${1\leq j\leq k_{i+1}}$.
  \end{enumerate}
\end{lemma}
\begin{proof}
  For the first item, we have ${F_{G,f}\models \sequ{1}{c^1}}$ so $c^1$ must be a fixed point if ${1\not\in\foind}$ and therefore is of type $1=\llayers(1)$ at each vertex because, on any other type of configuration, their is some vertex not in a state of type $1$ at which $F$ either applies projection $\pi$, or generates an error state. In case ${1\in\foind}$, $F^2(c^1)$ must be a fixed point for the same reason and therefore $c^1$ is of type $\llayers(1)=3$ at each vertex by definition of $F$ since no transition generating an error state can lead to $c^1$ and therefore only $\pi$ can be applied at each vertex.
  More generally, since ${F_{G,f}\models\sequple{i}{c^1,\ldots, c^i}}$ implies ${F_{G,f}\models\sequ{k}{c^{k-1},c^k}}$ for all ${2\leq k\leq i}$, it follows that configuration ${c^k}$ is of type $\llayers(k)$ at each vertex because $c^1$ lies in its orbit (is reached after precisely ${\llayers(k)-\llayers(1)}$ steps) and therefore no error state can appear in this orbit, so only $\pi$ is applied until reaching $c^1$.

  For the second item, consider configuration $c^l$ for ${l\in\foind}$. Since ${F_{G,f}\models\sequ{l}{c^{l-1},c^l}}$ if ${l\geq 2}$ or ${F_{G,f}\models\sequ{l}{c^1}}$ if $l=1$, we in particular have ${F_{G,f}\models\okfov{c^l}}$, so for each configuration $c$ such that ${F^2(c^l)=F^2(c)}$ it holds that $c$ has exactly one sibling configuration. We know that if $c'$ is $c$ or its sibling then it is a valid configuration (because there is no error state in its orbit which reaches the same fixed point as $c^l$), so we have ${K_l(c_v')\leq V_l^{\chi_l(c_v')}(c_v')}$ at each node $v$. We also have that ${\chi_l(c_v')}$ is uniform on the entire graph, because the graph is connected and thus if two vertices hold different values, then there would also be two \emph{neighboring} vertices holding different values which would contradict validity of the configuration. We deduce that there is exactly one vertex $v$ such that ${V_l^{\chi_l(c_v')}(c_v')=1}$ otherwise $c$ would not have exactly one sibling.

  Finally, for the third item, it is enough to complete configuration $c^i$ of type ${\llayers(i)}$ everywhere as shown above to a configuration $c^{i+1}$ of type $\llayers(i+1)$ everywhere which correctly encode assigned values in its $V_{i+1}$ layer (which is completely independent from $c^i$). It is enough to ensure ${F_{G,f}\models \sequ{i+1}{c^i,c^{i+1}}}$ when ${i+1\not\in\foind}$. In the case where ${i+1\in\foind}$, we have to additionally ensure that the choice layer ${\chi_{i+1}}$ is uniform and we can choose the control layer to be everywhere $0$ so that $c^{i+1}$ is indeed valid. Clearly ${F_{G,f}\models \okfov{c^{i+1}}}$ with this choices, so we have ${F_{G,f}\models \sequ{i+1}{c^i,c^{i+1}}}$. In both cases we deduce ${F_{G,f}\models \sequple{i+1}{c^1,\ldots ,c^{i+1}}}$, because by hypothesis we already have ${F_{G,f}\models \sequple{i}{c^1,\ldots ,c^{i}}}$.
\end{proof}

From the above lemma, if ${F_{G,f}\models\sequple{i}{c^1,\ldots, c^i}}$ then $c^i$ codes an assignment for all MSO variables ${x_l^j}$ for ${1\leq l\leq i}$ and ${1\leq j\leq k_l}$ through components ${V_l^j}$ of $c^i$. Precisely, when ${l\in\foind}$ variable ${x_l^j}$ is assigned to the unique vertex $v$ such that ${V_l^j(c^i_v)=1}$, and when ${l\not\in\foind}$ variable ${x_l^j}$ is assigned to the set of vertices ${\{v : V_l^j(c^i_v)=1\}}$. We denote this assignment by ${\alpha_{c^i}}$.
Moreover, ${\alpha_{c^i}}$ is an extension of assignment ${\alpha_{c^j}}$ for any ${1\leq j<i}$.

Next lemma ensures that ${\truth{\ldots}}$ predicate correctly codes truth of formula ${R(\ldots)}$ (matrix of $\Psi$) through the previous assignment translation.

\begin{lemma}\label{lem:fotruthcheck}
  Under the hypothesis of Lemma~\ref{lem:foquantifassign}, it holds that  ${F_{G,f}\models\truth{c^n}}$ if and only if ${(G,\alpha_{c^ n})\models R(x_1^1,\ldots, x_n^{k_n})}$.
\end{lemma}

\begin{proof}
  First, by Lemma~\ref{lem:foquantifassign}, $c^n$ is of type $\llayers(n)$ at each vertex and is a valid configuration. If $c$ is a pre-image of $c^n$, it must be valid (no error state in $c^n$) and of type $j$-truth-check for some $j$ for all vertices because the graph is connected (otherwise there would be two neighboring nodes with different types, contradicting local validity). It must also be the case that, at each vertex $v$, ${(\sigma,\mu)\locallymodels C_j}$ where ${\sigma=L(v)}$ and $\mu=\patt(c,v,1,1)$. Now consider any term $t$ of clause $C_j$ with first-order free variable $x_a^b$ and another free variable ${x_u^v}$, and let $v$ be the unique vertex such that ${V_a^b(c_v)=1}$, which is ${v=\alpha_{c^n}(x_a^b)}$. By definition of $f$, it must be the case that ${(G,\alpha_{c^n})\models t}$ (straightforward case analysis on atomic formulas). Since this holds for all terms of $C_j$ and since $R$ is a disjunction of clauses $C_j$, we deduce ${(G,\alpha_{c^ n})\models R(x_1^1,\ldots, x_n^{k_n})}$.

  Conversely, if ${(G,\alpha_{c^ n})\models R(x_1^1,\ldots, x_n^{k_n})}$, then by definition there must some $j$ such that ${(G,\alpha_{c^ n})\models C_j(x_1^1,\ldots, x_n^{k_n})}$. Let $c$ be the configuration of type $j$-truth-check everywhere such that ${\overline\pi(c)=c^n}$ (where $\overline\pi$ is the map applying $\pi$ at all vertices). We claim that ${F_{G,f}(c)=c^n}$ which implies ${F_{G,f}\models\truth{c^n}}$. Indeed, considering any term $t$ appearing in $C_j$ and any node $v$, and denoting $\sigma=L(v)$ and ${\mu=\patt(c,v,1,1)}$, we have:
  \begin{itemize}
  \item  ${(\sigma,\mu)\locallymodels t}$ if $v$ is the assignment of the leftmost free variable appearing in $t$ since ${(G,\alpha_{c^n})\models t}$,
  \item  ${(\sigma,\mu)\locallymodels t}$ by definition otherwise.
  \end{itemize}
  We deduce that ${(\sigma,\mu)\locallymodels C_j}$ and therefore that $F_{G,f}$ applies projection $\pi$ at each vertex $v$ on $c$, proving the claim.
\end{proof}

The proof of Theorem~\ref{theo:MSOFOCACON} then consists in applying inductively the definition of truth by assignments of variables simultaneously in $\Psi$ and $\phi$, use Lemma~\ref{lem:fotruthcheck} as base case and Lemma~\ref{lem:foquantifassign} for the induction step to translate assignments between MSO variables and FO variables.

\begin{proof}[Proof of Theorem~\ref{theo:MSOFOCACON}]
  Consider any ${(\Sigma,\Delta)}$-labeled graph which is connected. For ${1\leq i\leq n+1}$, denote by $\phi_i$ the subformula of $\phi$ starting from the $i$th quantifier (and without quantifier if ${i=n+1}$):
  \[\phi_i\isbydef \quantif_i'y^i\in\sequple{i}{y^1,\ldots, y^i},\ldots, \quantif_n'y^n\in\sequple{n}{y^1,\ldots,y^n}:\truth{y^n}\]
  and by $\Psi_i$ the subformula of $\Psi$ starting from the $i$th alternation of quantifiers:
  \[\Psi_i\isbydef \quantif_i x_i^1,\ldots, \quantif_i^{k_i}x_i^{k_i},\ldots,\quantif_n x_n^1,\ldots, \quantif_n x_n^n : R(x_1^1,\ldots, x_n^{k_n}).\]
  ${\phi_1=\phi}$ and ${\Psi_1=\Psi}$ and for ${i>1}$  the free variables of $\phi_i$ are ${y^1,\ldots, y^{i-1}}$ and those of $\Psi_i$ are ${x_1^1,\ldots, x_{i-1}^{k_{i-1}}}$.

  \framebox{Suppose first that ${F_{G,f}\models\phi}$.} We show by induction from ${i=n+1}$ downto $1$ that the following holds ($H_i$): 
  \begin{center}
    \begin{minipage}{.9\linewidth} \it
      for all configurations ${c^1,\ldots, c^{i-1}}$, if ${F_{G,f}\models\sequple{i-1}{c^1,\ldots, c^{i-1}}}$ and ${F_{G,f}\models\phi_i(c^1,\ldots, c^{i-1})}$ then ${(G,\alpha_{c^{i-1}})\models\Psi_i}$
    \end{minipage}
  \end{center}
  (to be understood without configuration and without assignment when ${i=1}$). Recall that assignment $\alpha_{c^{i-1}}$ is well-defined thanks to Lemma~\ref{lem:foquantifassign}. This implies ${G\models\Psi}$ since ${\phi_1=\phi}$ and ${\Psi_1=\Psi}$.
  \begin{itemize}
  \item the base case directly follows from Lemma~\ref{lem:fotruthcheck} since $\phi_{n+1}$ is exactly ${\truth{y^n}}$ and $\Psi_{n+1}$ is exactly ${R(x_1^1,\ldots, x_n^{k_n})}$;
  \item for the induction step, suppose that the hypothesis ($H_i$) holds and consider ${c^1,\ldots, c^{i-2}}$ such that ${F_{G,f}\models\sequple{i-2}{c^1,\ldots, c^{i-2}}}$ and ${F_{G,f}\models\phi_{i-1}(c^1,\ldots, c^{i-2})}$.
    \begin{itemize}
    \item if $\quantif'_{i-1} = \forall$, then ${F_{G,f}\models\phi_{i-1}(c^1,\ldots, c^{i-2})}$ means that for all configuration $c^{i-1}$ such that ${F_{G,f}\models\sequple{i-1}{c^1,\ldots, c^{i-1}}}$ it holds that ${F_{G,f}\models\phi_i(c^1,\ldots, c^{i-1})}$, and then by ($H_i$) it also holds that ${(G,\alpha_{c^{i-1}})\models\Psi_i}$. By item 3 of Lemma~\ref{lem:foquantifassign}, this implies ${(G,\alpha)\models\Psi_i}$ for any assignment $\alpha$ extending ${\alpha_{c^{i-2}}}$ to variables ${(x_{i-1}^1,\ldots, x_{i-1}^{k_{i-1}})}$. Now since ${\quantif_{i-1}}$ is a universal quantifier, this precisely means ${(G,\alpha_{c^{i-2}})\models\Psi_{i-1}}$.
    \item if $\quantif'_{i-1} = \exists$, then ${F_{G,f}\models\phi_{i-1}(c^1,\ldots, c^{i-2})}$ means that there exists a configuration $c^{i-1}$ such that ${F_{G,f}\models\sequple{i-1}{c^1,\ldots, c^{i-1}}}$ and ${F_{G,f}\models\phi_i(c^1,\ldots, c^{i-1})}$, and then by ($H_i$) it also holds that ${(G,\alpha_{c^{i-1}})\models\Psi_i}$. But ${\alpha_{c^{i-1}}}$ is some extension of assignment $\alpha_{c^{i-2}}$ to variables ${(x_{i-1}^1,\ldots, x_{i-1}^{k_{i-1}})}$. Thus, since $\quantif_{i-1}$ is an existential quantifier, we actually have ${(G,\alpha_{c^{i-2}})\models\Psi_{i-1}}$.
    \end{itemize}
  \end{itemize}

  \framebox{Suppose now that ${G\models\Psi}$.} We show by induction from ${i=n+1}$ downto $1$ that the following holds ($H_i'$): 
  \begin{center}
    \begin{minipage}{.9\linewidth} \it
      for any assignments $\alpha^{i-1}$ of variables ${x_1^1,\ldots,x_{i-1}^{k_{i-1}}}$, ${(G,\alpha^{i-1})\models\Psi_i}$ implies that for all configurations ${(c^1,\ldots, c^{i-1})}$ such that ${F_{G,f}\models\sequple{i-1}{c^1,\ldots, c^{i-1}}}$ and ${\alpha_{c^{i-1}}=\alpha^{i-1}}$, it holds that ${F_{G,f}\models\phi_i(c^1,\ldots,c^{i-1})}$
    \end{minipage}
  \end{center}
  (to be understood without free variable and without assignment when ${i=1}$). This implies ${(G\,F)\models\phi}$ since ${\phi_1=\phi}$ and ${\Psi_1=\Psi}$.
  \begin{itemize}
  \item the base case directly follows from Lemma~\ref{lem:fotruthcheck} since $\phi_{n+1}$ is exactly ${\truth{y^n}}$ and $\Psi_{n+1}$ is exactly ${R(x_1^1,\ldots, x_n^{k_n})}$.
  \item for the induction step, suppose that the hypothesis ($H_i'$) holds and consider an assignment $\alpha^{i-2}$ of variables ${x_1^1,\ldots,x_{i-1}^{k_{i-1}}}$ such that ${(G,\alpha^{i-2})\models\Psi_{i-1}}$
    \begin{itemize}
    \item if $\quantif_{i-1} = \forall$, then ${F_{G,f}\models\phi_{i-1}(c^1,\ldots, c^{i-2})}$ means that for all assignments $\alpha^{i-1}$ extending $\alpha^{i-2}$ to variables ${x_{i-1}^1,\ldots, x_{i-1}^{k_{i-1}}}$ we have ${(G,\alpha^{i-1})\models \Psi_i}$. So by hypothesis ($H_i'$) we also have ${F_{G,f}\models\phi_i(c^1,\ldots,c^{i-1})}$ for all configurations ${(c^1,\ldots, c^{i-1})}$ such that ${F_{G,f}\models\sequple{i-1}{c^1,\ldots, c^{i-1}}}$ and ${\alpha_{c^{i-1}}=\alpha^{i-1}}$. Since $\quantif_i'$ is universal with domain restriction by ${\sequple{i-1}{\ldots}}$, and because ${F_{G,f}\models\sequple{i-1}{c^1,\ldots, c^{i-1}}}$ implies ${F_{G,f}\models\sequple{i-2}{c^1,\ldots, c^{i-2}}}$, this actually means that, for all configurations ${(c^1,\ldots, c^{i-2})}$ such that ${F_{G,f}\models\sequple{i-2}{c^1,\ldots, c^{i-2}}}$ and ${\alpha_{c^{i-1}}=\alpha^{i-1}}$, it holds that  ${F_{G,f}\models\phi_{i-2}(c^1,\ldots, c^{i-2})}$. We thus have proven ($H_{i-1}'$).
    \item if $\quantif_{i-1} = \exists$, then ${F_{G,f}\models\phi_{i-1}(c^1,\ldots, c^{i-2})}$ means that there is an assignment $\alpha^{i-1}$ extending $\alpha^{i-2}$ to variables ${x_{i-1}^1,\ldots, x_{i-1}^{k_{i-1}}}$ with ${(G,\alpha^{i-1})\models \Psi_i}$. So, by hypothesis ($H_i'$), we also have ${F_{G,f}\models\phi_i(c^1,\ldots,c^{i-1})}$ for all configurations ${(c^1,\ldots, c^{i-1})}$ such that ${F_{G,f}\models\sequple{i-1}{c^1,\ldots, c^{i-1}}}$ and ${\alpha_{c^{i-1}}=\alpha^{i-1}}$. But for any configurations ${(c^1,\ldots, c^{i-2})}$ such that ${F_{G,f}\models\sequple{i-1}{c^1,\ldots, c^{i-2}}}$ there is a configuration $c^{i-1}$ such that ${F_{G,f}\models\phi_i(c^1,\ldots,c^{i-1})}$ and ${\alpha_{c^{i-1}}=\alpha^{i-1}}$ by item 3 of Lemma~\ref{lem:foquantifassign}. Since ${\quantif_{i-1}'}$ is existential with domain restriction by ${\sequple{i-1}{\ldots}}$ this actually means that for all configurations ${(c^1,\ldots, c^{i-2})}$ such that ${F_{G,f}\models\sequple{i-1}{c^1,\ldots, c^{i-2}}}$ it holds that ${F_{G,f}\models\phi_{i-1}(c^1,\ldots,c^{i-2})}$. We have thus proven ($H_{i-2}'$).
    \end{itemize}
  \end{itemize}
\end{proof}

\subparagraph*{Generalizing to arbitrary graphs.}
In the previous construction, we use the fact that considered graphs are connected in two places: to ensure uniformity of choice layers $\chi_i$ and to ensure a uniform choice of $j$ for testing clause $C_j$ with states of type $j$-truth-check. When generalized to possibly disconnected graphs, such uniformity conditions cannot be checked by the CA alone, simply because the CA has no possibility to communicate between connected components. We can compensate this impossibility by slightly changing the behavior of $F$ and adding new FO constraints in the definition of $\phi$. The price to pay is that the new definition of $\phi$ will depend on all parts of $\Psi$, not only its prefix signature.

First, the case of choice layers $\chi_i$ can be solved easily by de-grouping variables ${x_i^1}$ to ${x_i^{k_i}}$, \emph{i.e.} renumbering variables in the prefix of $\Psi$ by  letting ${k_i=1}$ and taking ${\sum_{1\leq i\leq n}k_i}$ as the new value of $n$. Then, choice layers become trivial (they contain just one state) and are therefore always uniform by definition.

To solve the case of truth check, we introduce a general pre-image counting trick to ensure that a configuration $c$ is uniform by a FO property. First, each state $s$ of the alphabet $S$ used by $c$ is associated to a distinct prime number $p_s$, and there is a probing mechanism that selects a set of vertices and allows exactly $p_s$ predecessors at selected vertices which are in state $s$, and only $1$ at vertices which are not selected. The trick to check that $c$ is $s$-uniform then consists in counting the number of pre-images up to ${\max_{s'\in S} p_{s'}}$: whatever the set of selected vertices, it should always be a power of $p_s$.

\begin{theorem}\label{theo:MSOFOCA}
  There are two recursive transformations $\tau_{FO}$ from MSO formulas to FO formulas and $\tau_{CA}$ from MSO formulas to CA local rules such that for any MSO formula $\Psi$, and any graph $G$ the following equivalence holds: ${G\models \Psi}$ if and only if ${F_{G,\tau_{CA}(\Psi)}\models\tau_{FO}(\Psi)}$.
\end{theorem}

Using notations from the construction of Theorem~\ref{theo:MSOFOCACON}, let us now describe precisely the modifications required to generalize to arbitrary graphs.
As explained above, we assume ${k_i=1}$ for all ${1\leq i\leq n}$, so choice layers $\chi_i$ are always trivial and uniform. States of type $l$ for ${1\leq l\leq\llayers(n)}$ are identical as in the construction of Theorem~\ref{theo:MSOFOCACON}.
However, we need additional states to implement the truth check for the matrix $R$ of formula $\Psi$, and it will spread over 3 time steps of the CA. The key is to ensure that some clause $C_j$ from the disjunctive normal form of $R$ is chosen uniformly on the entire graph and to check everywhere that each terms of $C_j$ is locally correct given the assignment. Let ${2=p_1< p_2<\cdots <p_d}$ be the first $d$ prime numbers.
For ${1\leq j\leq d}$, let ${T_j^0 = S_{\llayers(n)}\times\{j\}}$ and ${T_j^1 = T_j^0\times \{0,1\}}$ and ${T_j^2= T_j^1\times\{1,\ldots,p_j\}}$. $T_j^1$ is used to \emph{mark} vertices, while $T_j^2$ is used to alter the number of preimages depending on $j$ and the mark. We use the notation $\pi$ to denote at the same time the natural projection from ${T_j^0}$ onto $S_{\llayers(n)}$, or from ${T_j^1}$ onto ${T_j^0}$ or from ${T_j^2}$ onto ${T_j^1}$, which removes the rightmost component of states. ${T_j^0}$ sates will be used to check clause $C_j$, while states from ${T_j^1}$ and ${T_j^2}$ will be used to guarantee through a pre-image counting trick that the same choice of $j$ is made on the entire graph, thus ensuring correctness of the truth check of formula $R$.

We can now define the state set of the CA local rule ${\tau_{CA}(\Psi)}$ as
\[S = \{e_0,e_1\}\cup \bigcup_{1\leq l\leq \llayers(n)} S_l \cup \bigcup_{1\leq j\leq d} T_j^0\cup T_j^1\cup T_j^2.\]
We say that the \emph{type} of an element of ${S}$ is $l$ if it belongs to $S_l$, \emph{error} if it is $e_0$ or $e_1$ and \emph{$(j,m)$-truth-check} if it belongs to ${T_j^{m}}$ for ${1\leq j\leq d}$ and ${m=0,1}$ or $2$. 

A configuration ${c\in S^V}$ is \emph{valid} if the following conditions hold:
\begin{itemize}
\item states of all pairs of neighboring vertices of $G$ are of the same type, and not of error type;
\item at each vertex $v$ the control layers have zeros where the corresponding variable layers indicated by the choice layers have, precisely: ${K_i(c_v)\leq V_i^{\chi_i(c_v)}(c_v)}$ for all $i$ such that ${K_i(c_v)}$ is defined (intuitively, a $1$ in a control layer is authorized only if their is a $1$ in the 'chosen' component of the corresponding variable layer);
\item for a state ${(s,m,w)\in T_j^2}$ where $s\in T_j^0$, ${m\in\{0,1\}}$ and ${w\in\{1,\ldots,p_j\}}$, it must be the case that ${w=1}$ whenever ${m=0}$ (this condition expresses intuitively, that only marked vertices can generate preimages and it will allow through a preimage counting trick in the FO formula to ensure that the choice to check truth of clause $j$ is coherent on the entire graph).
 \end{itemize}
We write ${\locvalid(\mu)}$ when the capped multiset $\mu$ represents a locally valid neighborhood according to the above conditions.


The modified CA local rule ${f}$ is almost identical as the one from Theorem~\ref{theo:MSOFOCACON} and defined as follows (denoting again $s$ the unique state such that ${\mu(s,\epsilon)\geq 1}$):
\[f(\sigma,\mu) =
  \begin{cases}
    e_{1-i} &\text{ if $s=e_i$}\\
    e_0 &\text{ otherwise, and if $\neg\locvalid(\mu)$,}\\
    e_0 &\text{ otherwise, and if $s$ has type ${(j,0)}$-truth-check and ${(\sigma,\mu)\not\locallymodels C_j}$,}\\
    s &\text{ otherwise, and if $s$ is of type $1$,}\\
    \pi(s) &\text{ otherwise.}
  \end{cases}
\]

The key aspect of this new construction is that the correctness of the distributed truth check implemented in the CA above by states of type ${(j,m)}$-truth check rely on a modification of the considered FO formula. Let ${\preimj{j}{y}}$ be a FO formula expressing that the number of pre-images of $y$ is either ${>p_d}$ or a multiple of $p_j$. We use the following  modified definition of formula $\truth{y}$:
\[\truth{y}\isbydef \bigvee_{1\leq j\leq d}\exists y^j : y^j\to y \wedge \bigl(\forall y': y'\to y^j \Rightarrow \preimj{j}{y'}\bigr)\]
which intuitively makes sense when $y$ represents a configuration everywhere of type $\llayers(n)$, so $y^j$ represents of configuration of type ${(j,0)}$-truth check everywhere. We this new definition of $\truth{y}$, we introduce a strong dependence of the FO-formula on the matrix part of $\Psi$ (by the presence of $d$ for instance), which was not the case in Theorem~\ref{theo:MSOFOCACON}.

The final FO formula $\phi$ is defined exactly as in Theorem~\ref{theo:MSOFOCACON} but using this modified version of $\truth{y}$:

\[\phi\isbydef \quantif_1' y^1\in\sequple{1}{y^1}, \quantif_2' y^2\in\sequple{2}{y^1,y^2},\ldots, \quantif_n' y^n\in\sequple{n}{y^1,\ldots,y^n} : \truth{y^n}\]

where the FO quantifier $\quantif_i'$ is $\exists$ if the MSO quantifier $\quantif_i$ is existential, and $\forall$ if $\quantif_i$ is universal.\\

The proof of Theorem~\ref{theo:MSOFOCA} is based on two lemmas adapted from Lemma~\ref{lem:foquantifassign} and \ref{lem:fotruthcheck}, and can be copied word for word from the proof of Theorem~\ref{theo:MSOFOCACON}, but simply removing the connectedness hypothesis.

\begin{lemma}\label{lem:foquantifassignnotcon}
  Let $G$ be a ${(\Sigma,\Delta)}$-labeled graph. Consider configurations ${c^1,\ldots, c^i}$ with ${1\leq i\leq n}$ such that ${F_{G,f}\models \sequple{i}{c^1,\ldots, c^i}}$, then the following holds:
  \begin{enumerate}
  \item $c^k$ is of type $\llayers(k)$ at each vertex, for ${1\leq k\leq i}$;
  \item  for ${1\leq l\leq i}$ and ${1\leq j\leq k_l}$, variable component $V_l^j(c^k)$ is the same for for all $k$ with ${l\leq k\leq i}$, and it is such that exactly one vertex is in state $1$ when $l\in\foind$ ;
  \item if ${i<n}$, for any assignment $\alpha$ of variables ${(x_{i+1}^1,\ldots,x_{i+1}^{k_{i+1}})}$ there exists $c^{i+1}$ such that \[F_{G,f}\models \sequple{i+1}{c^1,\ldots ,c^{i+1}}\] and ${V_{i+1}^j(c^{i+1}) = \alpha(x_{i+1}^j)}$ for all ${1\leq j\leq k_{i+1}}$.
  \end{enumerate}
\end{lemma}
\begin{proof}
  Straightforward adaptation from the proof of Lemma~\ref{lem:foquantifassign}, where $\chi_i$ components are by definition uniform since $k_i=1$ so that we don't use the hypothesis of connectivity on the graph.
\end{proof}



\begin{lemma}\label{lem:fotruthchecknotcon}
  Under the hypothesis of Lemma~\ref{lem:foquantifassignnotcon}, it holds that  ${F_{G,f}\models\truth{c^n}}$ if and only if ${(G,\alpha_{c^ n})\models R(x_1^1,\ldots, x_n^{k_n})}$.
\end{lemma}
\begin{proof}
  First, by Lemma~\ref{lem:foquantifassignnotcon}, the hypothesis implies that $c^n$ is of type $\llayers(n)$ at each vertex and is a valid configuration. If $c$ is a pre-image of $c^n$, it must be valid (because no error state appears in $c^n$), as well as any of its preimage, and any preimage of its preimage.\\

Suppose now that ${F_{G,f}\models\truth{c^n}}$.
We claim that any pre-image $c$ of $c^n$ must be of type ${(j,0)}$-truth-check at every vertex for some $j$ (\textit{i.e.} of uniform type on the entire graph). Indeed, if it were not the case, there would exist two vertices $v$ and $v'$ such that $c$ is of type ${(j,0)}$-truth-check at $v$ and of type ${(j',0)}$-truth-check at $v'$ with ${j\neq j'}$. Thus $c$ would have a preimage $c'$ where only $v$ is marked which would have exactly ${p_j}$ preimages, and another $c''$ where only $v'$ is marked which would have exactly ${p_{j'}}$ preimages, thus contradicting ${F_{G,f}\models\truth{c^n}}$.

It must also be the case that, at each vertex $v$, ${(\sigma,\mu)\locallymodels C_j}$ where ${\sigma=L(v)}$ and $\mu=\patt(c,v,1,1)$, because $F_{G,f}$ applied on $c$ checks this condition and $c^n$ contains no error. Now consider any term $t$ of clause $C_j$ with first-order free variable $x_a^b$ and another free variable ${x_u^v}$, and let $v$ be the unique vertex such that ${V_a^b(c_v)=1}$, which is ${v=\alpha_{c^n}(x_a^b)}$. By definition of $f$, it must be the case that ${(G,\alpha_{c^n})\models t}$ (straightforward case analysis on atomic formulas). Since this holds for all terms of $C_j$ and since $R$ is a disjunction of clauses $C_j$, we deduce ${(G,\alpha_{c^ n})\models R(x_1^1,\ldots, x_n^{k_n})}$.\\

  Conversely, if ${(G,\alpha_{c^ n})\models R(x_1^1,\ldots, x_n^{k_n})}$, then by definition there must some $j$ such that ${(G,\alpha_{c^ n})\models C_j(x_1^1,\ldots, x_n^{k_n})}$. Let $c$ be the configuration of type $(j,0)$-truth-check everywhere such that ${\overline\pi(c)=c^n}$ (where $\overline\pi$ is the map applying $\pi$ at all vertices). We claim that ${F_{G,f}(c)=c^n}$ which implies ${F_{G,f}\models\truth{c^n}}$ (because clearly all preimages of $c$ have a number of preimages which is a power of $p_j$ by construction). Indeed, considering any term $t$ appearing in $C_j$ and any node $v$, and denoting $\sigma=L(v)$ and ${\mu=\patt(c,v,1,1)}$, we have:
  \begin{itemize}
  \item  ${(\sigma,\mu)\locallymodels t}$ if $v$ is the assignment of the leftmost free variable appearing in $t$ since ${(G,\alpha_{c^n})\models t}$,
  \item  ${(\sigma,\mu)\locallymodels t}$ by definition otherwise.
  \end{itemize}
  We deduce that ${(\sigma,\mu)\locallymodels C_j}$ and therefore that $F_{G,f}$ applies projection $\pi$ at each vertex $v$ on $c$, proving the claim.
\end{proof}

\section{Consequences on FO model checking for CA}
\label{sec:coros}

A set of ${(\Sigma,\Delta)}$-labeled graphs, or graph language, is \emph{MSO-definable} if it is of the form ${\graphs{\Psi}}$ for some MSO formula $\Psi$. It is \emph{FOCA-definable} if it is of the form ${\graphs{\phi,f}}$ for some FO formula $\phi$ and some CA local rule $f$. From Theorem~\ref{theo:MSOFOCA}, we get the following immediate corollary.

\begin{corollary}
  MSO-definable and FOCA-definable graph languages are the same.
\end{corollary}

Since the translations given in Section~\ref{sec:translation} are effective, we also obtain equivalence of model checking problems. Using \cite{Kuske_2005}, this gives a characterization of decidability of FO model checking for CA orbits on f.g. groups.

\begin{corollary}\label{coro:focamodelchecking}
  On any fixed graph, FO model checking for CA is many-one equivalent to MSO model checking.
  In particular, FO-model checking for CA on a f.g. group $\Gamma$ is decidable if and only if $\Gamma$ is virtually free.
\end{corollary}

We can get a more precise result on graphs of bounded degree, but we need an additional lemma in order to apply Theorem~\ref{theo:MSOFOCACON}.
It is well-known that undecidability in MSO can be obtained using the MSO-definability of grids and encoding Turing-computations on it \cite{courcellebook}. Moreover, MSO on bounded degree graphs actually allows to code quantification on edge sets and not only vertex sets \cite{Courcelle_1994}, so that we can express the grid minor relation (and not only the fact that a fixed graph is a minor). The following lemma doesn't use any new idea, but it ensures the fact that all this encoding process can be done within a fixed prefix signature. It also gives a variant of the construction for $\Sigma_1^1$-hardness (see \cite{rogers1987theory}) using infinite grids and the recurring domino problem \cite{Harel_1985}.

  A ${n\times n}$-grid (or simply a grid when $n$ is not specified), is the directed finite graph with vertices ${\{(i,j) : 1\leq i,j\leq n\}}$ and edges the set of pairs ${\bigl((i,j),(i+1,j)\bigr)}$ (the east edges) for ${1\leq j\leq n}$ and ${1\leq i < n}$, and pairs ${\bigl((i,j),(i,j+1)\bigr)}$ (the north edges) for ${1\leq j< n}$ and ${1\leq i \leq n}$.
A $\infty$-grid is the infinite directed graph with vertices ${\N\times\N}$ and same north/east adjacency relation.

\begin{lemma}\label{lem:notarithmso}
  Fix some $D$.  There exists a fixed quantifier prefix signature $\rho$ for MSO such that, for any graph $G$ of degree at most $D$ that contains arbitrarily large grid as minors, deciding whether a given MSO formula in prenex form with prefix signature $\rho$ is satisfied on $G$ is undecidable.
  There also exists a fixed quantifier prefix signature $\rho'$ for MSO such that, for any graph $G$ of degree at most $D$ that contains a $\infty$-grid as minor, deciding whether a given MSO formula in prenex form with prefix signature $\rho'$ is satisfied on $G$ is $\Sigma_1^1$-hard.
\end{lemma}

\newcommand\psiadj{\Psi_{\texttt{adj}}}
\newcommand\gridnorth{\Psi_{\texttt{north}}}
\newcommand\grideast{\Psi_{\texttt{east}}}
\newcommand\psigrid{\Psi_{\texttt{grid}}}
\newcommand\psiminor{\Psi_{\texttt{minor}}}
\newcommand\griduple{\overline{X_{\mathtt{grid}}}}

\begin{proof}
  In this proof, label sets $\Sigma$ and $\Delta$ play no role as we will work on the underlying undirected unlabeled graph of $G$. So to simplify the proof, let's suppose that $G$ is undirected and that $\Sigma$ and $\Delta$ are singletons.
The necessity to have orientation inside grids to encode computation is recovered from the undirected graph $G$ through MSO by introducing adequate colorings. More precisely, an oriented grid is represented by a set of vertices $X$, a formula ${\psiadj(x_1,x_2)}$ substituted for an undirected adjacency relation on $X$, and a coloring on $X$, \textit{i.e.} a tuple of sets ${\griduple}$.
We then have an MSO formula ${\psigrid}$ such that ${\psigrid(X,\psiadj,\griduple)}$ holds if and only if the graph represented by $X$ and $\psiadj$ and $\griduple$ is isomorphic to an oriented grid. We also have formulas to test oriented adjacency along the direction \texttt{north} and \texttt{east}, for instance ${\gridnorth(X,\psiadj,\griduple,x_1,x_2)}$ holds if and only if $x_2$ is at the north of $x_1$ in the grid graph represented by $X$, $\psiadj$ and $\griduple$. We can also identify each side of the grid with MSO formulas. 

  The minor relation is definable in MSO2, the extension of MSO allowing quantification over vertices, set of vertices but also sets of edges. Indeed, given a set of vertices $X$ and a set of contracted edges $Y$ and a set of suppressed edges $Z$, we can write the property that two vertices $x_1$ and $x_2$ are neighbors in the minor graph induced by vertices $X$ with contraction of edges $Y$ and suppression of edges $Z$ by
expressing that $x_i$ is connected to some $x_i'$ through edges from $Y$ (for ${i=1,2}$) and that $x'_1$ is connected to $x'_2$ by an edge not in $Z$ nor in $Y$.
  
On bounded degree graphs, MSO2 definability can actually be translated into MSO definability \cite{Courcelle_1994} (see \cite[Theorem 7.10 and Theorem 9.38]{courcellebook} for more general results). The intuition is that it is possible to express in MSO that some coloring on vertices has strong enough properties so that an edge can be identified by a vertex and a color. This way a set of edges can be coded as a tuple of set of vertices (the size of the tuple depends on the degree of the graph). We therefore have an MSO formula $\psiminor(X,\overline{X_C},\overline{X_D},x_1,x_2)$ expressing that $x_1$ and $x_2$ are neighbors in the minor graph of $G$ induced by set of vertices $X$ and with contraction of edges coded via the tuple ${\overline{X_C}}$ and deletion of edges coded via the tuple ${\overline{X_D}}$.

Given a finite set $S$ and sets of horizontal and vertical dominos ${D_{\mathtt{east}}, D_{\mathtt{north}}\subset S\times S}$ we can express that a coloring of $X$ by $S$ represented through ${\overline{X_S}}$ respects the horizontal (resp. vertical) domino constraints on a pair $x_1$, $x_2$ of vertices by a simple disjunction of case on $D_{\mathtt{east}}$ (resp. $D_{\mathtt{north}}$) denoted $\Psi_{h}(x_1,x_2,\overline{X_S})$ (resp. ${\Psi_{v}(x_1,x_2,\overline{X_S})}$).  We also denote by ${\Psi_{OK}(X,\overline{X_S})}$ the formula expressing that $\overline{X_S}$ is a partition of $X$.
It is well-known that this formalism for the domino problem is equivalent to the Wang tile formalism through higher-block recoding (see proof of Theorem~\ref{theo:dominofoca} for details).

With this tools, we can first prove the first part of the lemma.
  Let us define formula $\Psi$ (depending on some given domino constraints) which expresses that any grid minor of $G$ admits a coloring by $S$ respecting the horizontal and vertical domino constraints, precisely:
  \begin{align*}
    \Psi &\isbydef \forall X, \forall \overline{X_C}, \forall\overline{X_D}, \forall \griduple,  \psigrid\bigl(X,\psiminor(X,\overline{X_C},\overline{X_D},x_1,x_2),\griduple\bigr) \Rightarrow \\
    & \mathbf{\exists \overline{X_S}}, \Psi_{OK}(X,\overline{\mathbf{X_S}})\wedge\forall x_1,\forall x_2, \\
    &\gridnorth\bigl(X,\psiminor(X,\overline{X_C},\overline{X_D},x_1,x_2),\griduple,x_1,x_2\bigr)\Rightarrow\Psi_{v}(x_1,x_2,\overline{\mathbf{X_S}})\\
&\wedge \grideast\bigl(X,\psiminor(X,\overline{X_C},\overline{X_D},x_1,x_2),\griduple,x_1,x_2\bigr)\Rightarrow\Psi_{h}(x_1,x_2,\overline{\mathbf{X_S}})
  \end{align*}

  Given a set of domino constraints it is undecidable whether ${G\models\Psi}$ because it is equivalent to the domino problem \cite{berger}: $G$ has arbitrarily large grid minors and having correct coloring for arbitrarily large grids as expressed by $\Psi$ is equivalent to having a correct coloring for an infinite grid by compacity.

  Moreover, applying the standard logical transformation for the first implication, it can be verified that $\Psi$ has a prenex normal form starting by ${\forall X, \forall \overline{X_C}, \forall\overline{X_D}, \forall \griduple, \forall \overline{X_S}}$ followed by a sequence of quantifiers which does not depend on $S$ nor on the domino constraints. This means that we can write the prenex normal form of $\Psi$ with a fixed prefix signature when $S$ and the domino constraints vary.\\

  Now, for the second part of the lemma, we use a reduction from a variant of the recurring domino problem which gives $\Sigma_1^1$-hardness (using the horizontal/vertical domino formulation instead of Wang tiles). Precisely, the following problem is $\Sigma_1^1$-hard for a fixed $s_0$ (problem R2 of \cite{Harel_1985}): deciding whether, for a given set of horizontal/vertical domino constraints with ${s_0\in S}$, there is a coloring of the $\infty$-grid that satisfy the constraints and contains infinitely many occurrences of $s_0$ on the leftmost column.

  To encode this problem, we first use an MSO formula $\Psi_\infty(X,\psiadj,\overline{X_S})$ expressing that:
  \begin{enumerate}
  \item the graph ${G(X,\psiadj)}$ represented by $X$ and $\psiadj$ has a unique vertex of in-degree $0$ denoted $x_0$;
  \item for any induced (finite) grid in the graph ${G(X,\psiadj)}$ with lower-left corner $x_0$, there is a strictly larger (finite) grid with the same lower-left corner such that:
    \begin{itemize}
    \item $\overline{X_S}$ is a partition and satisfies domino constraints on the larger grid, and
    \item there is a vertex on the leftmost column of the larger grid, outside the smaller grid, which is colored in $s_0$ by the coloring $\overline{X_S}$.
  \end{itemize}
  \end{enumerate}

  We then consider the formula $\Psi$ (depending on the domino constraints) which express that there exists a minor of $G$ and a coloring of it that satisfies $\Psi_\infty$:
  \[\Psi\isbydef \exists X, \exists \overline{X_C}, \exists\overline{X_D}, \exists\overline{X_S} : \Psi_\infty(X,\psiminor(X,\overline{X_C},\overline{X_D},x_1,x_2),\overline{X_S}).\]
  First, it should be clear that $\Psi$ has a prenex normal form where the prefix signature doesn't change when $S$ and the domino constraints vary.
  Second, we claim that $G\models\Psi$ if and only if there exists a coloring of the $\infty$-grid satisfying the domino constraints and containing infinitely many occurrences os $s_0$ on the first column. Indeed, if there is such a coloring of the $\infty$-grid and since the $\infty$-grid occurs as a minor of $G$, we can chose the corresponding variables quantified existentially to represent exactly this minor and this coloring so that $\Psi_\infty$ will be satisfied. Conversely, if ${(G,\alpha)\models \Psi_\infty(X,\psiminor(X,\overline{X_C},\overline{X_D},x_1,x_2),\overline{X_S})}$ for some assignment $\alpha$ of the variables representing both a minor and a coloring, then we can extract a valid coloring of the $\infty$-grid with infinitely many $s_0$ on the leftmost column from the validity of $\Psi_\infty$: just consider coloring $\overline{X_S}$ on the sequence of larger and larger grids which extend each other and add at each step a new occurrence of $s_0$ on the leftmost column.
\end{proof}

This lemma together with Theorem~\ref{theo:MSOFOCACON} gives undecidability of model checking for a fixed FO formula expressed in the following corollary. For the first part of the corollary, we use the grid minor theorem \cite{Robertson_1986} to translate the statement of the lemma in terms of treewidth \cite{Robertson_treewidth}, and for the second part, it is known that the Cayley graph of a f.g. group which is not virtually free has a thick end (see \cite{Ballier_2018}), hence it contains a $\infty$-grid as a minor by Halin's grid theorem \cite{Halin_1964}. The decidable part of the Corollary comes from \cite{Courcelle_1988,Muller_1985}.

\begin{corollary}\label{coro:onefo}
  Fix some $D$. There is a fixed FO formula $\phi$ such that for any connected graph $G$ of degree at most $D$, the set ${\cas{G,\phi}}$ is computable if and only if $G$ has finite treewidth. Moreover,   there is a FO formula $\phi'$ such that for any Cayley graph $G$ of a f.g. group with at most $D$ generators, the set ${\cas{G,\phi}}$ is:
  \begin{itemize}
  \item computable if the group is virtually free,
  \item $\Sigma_1^1$-hard otherwise.
  \end{itemize}

\end{corollary}


In the context of modeling, it makes sense to consider distributed dynamical systems over arbitrary finite graphs. For instance automata networks (which are non-uniform CA on arbitrary finite graphs) are a well-established model for its use in the study of gene regulation networks \cite{Thomas_1973,Kauffman_1969}. The theory of automata networks has largely grown around FO properties of orbits (typically fixed points) and their crucial dependence on the graph \cite{Richard_2018}. However, although many results deal with computational complexity in automata networks \cite{Bridoux_2022,ricean}, no natural undecidability result appeared so far to our knowledge. By our translation result from MSO, we can import Trakhtenbrot's theorem \cite{Ebbinghaus_1995} to obtain an undecidability result for FO properties of CA orbits on finite graphs. It can intuitively be formulated as follows in the context of modeling: it is undecidable to know whether there is some finite interaction graph on which a given local interaction law (CA) induces a given dynamical property (FO).
By the way, this corollary doesn't need Trakhtenbrot's theorem since we have all the expressive power of MSO (not only FO logic on graphs), and we can obtain it for a fixed FO formula. For instance, it follows directly from Theorem~\ref{theo:MSOFOCACON} and the techniques of Lemma~\ref{lem:notarithmso} (see also \cite[Theorem 5.6]{courcellebook}).

\begin{corollary}\label{coro:finitesatfoca}
  There exists a FO formula $\phi$ such that the following problem is undecidable:  given some input CA local rule $f$, decide whether there exists a finite graph $G$ with ${F_{G,f}\models \phi}$.
\end{corollary}

\begin{proof}
  Finite rectangular ${n\times m}$-grids are MSO-definable \cite[proof of Proposition 5.14]{courcellebook} and their lower-left and top-right corners are also MSO-identifiable. As in lemma~\ref{lem:notarithmso}, they are represented through a coloring $\griduple$ and a set of vertices $X$. We can also easily express in MSO that some Turing computation starts from the lower-left corner on an empty tape and finishes in an accepting state exactly at the upper-right corner. Formally, a Turing machine computation is a coloring ${\overline{X_S}}$ of the grid respecting vertical and horizontal domino constraints. All this conditions being gathered in a formula $\Psi_{\mathtt{halt}}(X,\griduple,\overline{X_S})$, we can define formula $\Psi$ expressing that the entire graph is a grid that can hold a correct halting computation by:
  \[\exists X,\exists \griduple, \exists\overline{X_S}, \Psi_{\mathtt{halt}}(X,\griduple,\overline{X_S})\wedge \forall x, x\in X.\]
  Since no quantification depending on $S$ and the Turing computation appears in $\Psi_{\mathtt{halt}}$, it is clear that $\Psi$ can be written in prenex normal form with a fixed prefix signature when the Turing machine varies. So we can apply Theorem~\ref{theo:MSOFOCACON} and effectively obtain a pair $\phi,f$ such that ${F_{G,f}\models\phi}$ if and only if ${G\models\Psi}$ and $\phi$ doesn't depend on the Turing machine considered. $\phi$ is the formula announced in the corollary, and it is clear  that ${G\models\Psi}$ for some finite graph $G$ if and only if the Turing machine encoded in $\Psi$ halts starting from the empty tape, and halts exactly at the rightmost ever visited position of its tape. It remains to check that the halting problem with the additional condition of halting at the rightmost visited position of the tape is undecidable, which is straightforward by transforming any machine $M$ (with a tape alphabet using a blank state which can only be erased and never be written) to a new machine working as $M$ but, as soon as $M$ halts, launches a terminal routing that moves the head to the right until reaching the first blank state before stopping.
\end{proof}

\section{Cayley graphs and domino problems}
\label{sec:cayley-graphs-domino}

\newcommand\asympto{\overset{\infty}{=}}

The case of Cayley graphs of f.g. groups is particular for our approach in two relevant ways.
First, we can express in MSO that a set of vertices is infinite.
\newcommand\uplewalk{\overline{X_{\mathtt{walk}}}}

\begin{lemma}\label{lem:msocayleyinfinite}
  For any $D\geq 1$, there is an MSO formula ${\Psi(X)}$ such that, on any Cayley graph $G$ of some f.g. group with $D$ generators and any assignment $\alpha$, it holds that ${(G,\alpha)\models \Psi(X)}$ if and only if ${\alpha(X)}$ is infinite.
\end{lemma}

\begin{proof}
  First, for any infinite set of vertices $X$, there is an oriented walk on the group passing at most $2$ times by any vertex and visiting infinitely many elements of $X$. To see this, consider any spanning tree $T$ of the Cayley graph rooted at an arbitrary vertex $v_0$. At least one of the subtrees of $v_0$ contains infinitely many elements from $X$ (finite degree of the Cayley graph): choose one and choose one ${x\in X}$ belonging to this subtree. Then let $\rho$ be the path from $v_0$ to $x$ in $T$, and let $v$ be the last vertex of this path which possesses a subtree $T_v$ with infinitely many elements from $X$. Note that $v$ can be $v_0$ or $x$ or another vertex, but in any case the subtree $T_v$ is disjoint from $\rho$. We define the start of the walk as follows: $\rho$ from $v_0$ to $x$, then back from $x$ to $v$ following the inverse of $\rho$. We can then iterate this reasoning inside subtree $T_v$ from $v$ to obtain the desired path.

  To define a walk passing at most $2$ times by any vertex, it is sufficient to code that each vertex contains $0$, $1$ or $2$ positions (corresponding to steps of the walk) and associate to each position an outgoing direction $\delta\in\Delta$ and a number $1$ or $2$ indicating which of the two positions to go to in the vertex pointed by direction $\delta$. Such a walk can be represented by a tuple of second-order variables $\uplewalk$, and we can express by a formula $\Psi_{\mathtt{OK}}(\uplewalk)$ that this tuple actually represents a valid walk: each vertex contains at most $2$ positions, each position has a successor, exactly one position has no predecessor and all others have exactly one, and the walk is connected. We can also express by a formula ${\Psi_{\mathtt{reach}}(x_1,x_2,\uplewalk)}$ that vertex $x_2$ is reached by the walk $\uplewalk$ starting from $x_1$, or $x_1$ is outside the walk. From this, we have that the following formula $\Psi(X)$ expresses that $X$ is infinite:
  \[\Psi(X)\isbydef\exists\uplewalk,\Psi_{\mathtt{OK}}(\uplewalk)\wedge\forall x_1, \exists x_2: x_2\in X \wedge \Psi_{\mathtt{reach}}(x_1,x_2,\uplewalk).\]
It can be check that this formula only depends on $\Delta$, so it works on any Cayley graph of f.g. groups with generators renamed $\Delta$.
\end{proof}

\newcommand\FOasy{FO($\asympto$)}
From this lemma, it makes sense to extend the signature of FO logic with the addition of new relation $\asympto$ on configurations which is at the heart of the 'Garden of Eden' theorem \cite{cagroups}: we write ${c\asympto c'}$ whenever ${\{v:c_v\neq c'_v\}}$ is finite. We denote by \FOasy{} the extension of FO signature by adding relation $\asympto$. 
By Lemma~\ref{lem:msocayleyinfinite}, this extension remains within MSO. Precisely, in any fixed Cayley graph of a f.g. group and
by a straightforward extension of Theorem~\ref{theo:FOCAtoMSO}, we can compute from any formula in \FOasy{} and CA local rule, an equivalent MSO formula. We deduce that \FOasy{} model checking for CA is decidable on some f.g. group exactly when MSO model checking is, and exactly when FO model checking for CA is.

Besides, if $\Gamma$ is a f.g. group and $G_1$ and $G_2$ two Cayley graphs of $\Gamma$ with two different sets of generators $\Delta_1$ and $\Delta_2$, then the definable CA global maps ${F: S^\Gamma\to S^\Gamma}$ are the same on $G_1$ and $G_2$. More precisely, there is a computable translation $\tau$ on CA local maps such that for any local map $f$ for $G_1$, it holds:
${F_{G_1,f} = F_{G_2,\tau(f)}}$.
This simply comes from the fact that we can translate $\Delta_1$-walks into equivalent $\Delta_2$-walks.

Therefore, if $\phi$ is a FO formula, we have that the sets $\cas{\phi,G_1}$ and $\cas{\phi,G_2}$ are actually Turing-equivalent. Said differently, by Theorem~\ref{theo:FOCAtoMSO}, any FO formula (actually any \FOasy{} formula) defines a fragment of MSO logic whose model checking problem's Turing degree is independent of the choice of generators on a f.g. group $\Gamma$. It turns out that such fragments naturally capture the domino problem and its classical variants.

Given some finite set $S$, a domino specification $\domino$ is a set of pairs ${D_\delta\subseteq S^2}$ for each ${\delta\in\Delta}$. A configuration ${c\in S^V}$ is said $\domino$-valid for some graph if for any ${v,v'\in V}$ it holds: ${(v,v')\in E_\delta}$ implies ${(c_v,c_{v'})\in R_\delta}$. The domino problem on a fixed graph, consists in deciding given $\domino$ whether there exists a $\domino$-valid configuration. The seeded domino problem consists in deciding given ${\domino}$ and ${s_0\in S}$ whether there exists a $\domino$-valid configuration where $s_0$ occurs at some vertex. Finally, the recurring domino problem consists in deciding given ${\domino}$ and ${s_0\in S}$ whether there exists a $\domino$-valid configuration where $s_0$ occurs infinitely often. 


\begin{theorem}\label{theo:dominofoca}
  Fix any Cayley graph $G$ of any f.g. group, then:
  \begin{itemize}
  \item domino problem $\turingequiv$  ${\exists x, x\to x}$, 
  \item seeded domino problem $\turingequiv$    ${\exists x, \exists y, x\to x\wedge y\to x\wedge x\neq y}$, 
  \item recurring domino problem $\turingequiv$ ${\exists x, \exists y, x\to x\wedge y\to x\wedge \neg (x\asympto y)}$,
  \end{itemize}
  where '${\turingequiv \phi}$' means Turing-equivalent to the set ${\cas{\phi,G}}$ (model checking of $\phi$ for CAs on $G$).
\end{theorem}

\begin{proof}
  In this proof $\Sigma$ is a singleton and therefore omitted from CA local rules to simplify notations.
  
  For the first item, given some domino specification $\domino$ it is easy to define a CA local rule $f_\domino$ whose fixed points are exactly $\domino$-valid configurations:
  \[f_\domino(\mu) =
    \begin{cases}
      s &\text{ if }\forall\delta, \mu(s',\delta)>0\Rightarrow (s,s')\in D_\delta\\
      s'\neq s&\text{ otherwise}
    \end{cases}
  \]
  where $s$ denotes the unique state such that ${\mu(s,\epsilon)>0}$.

 Conversely, given a local rule $f$ of radius $r$ over state set $S$ , we can define a domino specification $\domino$ over some state set $S' \subseteq S^{B(r)}$ where $B(r)$ denotes the ball of radius $r$ in the graph. It is just a matter of applying a higher-block recoding (a well-known technique in symbolic dynamics) and restrict to local patterns of fixed points. Let $S'$ be the set of patterns on which the local rule $f$ doesn't change the state (recall that on Cayley graphs of f.g. groups, $\mu$ gives the same information as a pattern from $S^{B(r)}$). Then the domino specification is defined by
  \[(p,p')\in D_\delta \iff \forall v,v'\in B(r), (v,v')\in E_\delta\Rightarrow p_{v}=p'_{v'}\]
  expressing that moving in direction $\delta$ can only change the state in a way compatible with translation of local patterns, or said differently, that the $S'$ configuration is the correct higher-block recoding of some $S$ configuration. It is clear that a $\domino$-valid configuration is just a higher-block recoding of a fixed-point of $f$.\\

  For the second and third items, given $\domino$ over state set $S$ and ${s_0\in S}$, we define $f_{\domino,s_0}$  over state set ${S\cup\{t,e_0,e_1\}}$ (where ${t,e_0,e_1}$ are not in $S$), which duplicates state $s_0$ in another state $t$ and use error states $e_0$ and $e_1$. Precisely:
  \[f_{\domino,s_0}(\mu) = 
    \begin{cases}
      e_{1-i} &\text{ if $s=e_i$}\\
      s_0 &\text{ otherwise and if $s=t$},\\
      s &\text{ otherwise and if }\forall\delta, \mu(s',\delta)>0\Rightarrow (\rho(s),\rho(s'))\in D_\delta,\\
      e_0&\text{ otherwise.}
    \end{cases}
  \]
  where $\rho : S\cup\{t\}\to S$ is the map that sends $t$ to $s_0$ and is the identity on $S$.
  Intuitively, this rule checks domino constraints while considering $t$ equal to $s_0$ and transforms any occurrence of $t$ to $s_0$. If a domino constraint is violated somewhere it generates an error state $e_i$ which will locally oscillate with period two thus preventing forever the existence of a fixed point.
  One can check that fixed point for $f_{\domino,s_0}$ are exactly the $\domino$-valid configuration: it cannot contain state $t$, nor any error state, and must follow the $D_\delta$ constraints everywhere. Moreover, such a configuration possesses a preimage other than itself if and only if it contains an occurrence of $s_0$ which is replaced by $t$ in the preimage. Moreover, each vertex at which such a configuration differs from its preimage must be in state $s_0$ in the configuration and $t$ in the preimage. Thus $\domino,s_0$ is a positive instance of the seeded (resp. recurring) domino problem if and only if ${f_{\domino,s_0}}$ satisfies formula of item two (resp. item three). 

  Finally, consider some CA local rule $f$ over state set $S$ and define a new set of domino constraints $\domino^+$ over state set ${S'\subseteq S^{B(r)}\times S^{B(r)}}$ that uses higher-block recoding as above, but this time to represent pairs of configuration ${(c,c')}$ such that $c$ is a fixed-point and $c'$ is a pre-image of $c$. Here $S'$ is the set of pairs of patterns over domain $B(r)$ such that the first pattern of the pair induces no change of state at position $0$ by local rule $f$ (\textit{i.e.}, locally a fixed point) and the second pattern of the pair is such that its image under $f$ is the state at position $0$ of the first (\textit{i.e.}, locally the second configuration is a preimage of the first). Then define $S'_0$ as the subset of $S'$ where the state at position $0$ differs between the first and second pattern. It then holds that $\domino^+$-valid configurations are exactly the higher-block recodings of pairs of configurations ${(c,c')}$ where $c$ is a fixed-point, $c'$ is a pre-image of $c$ and, moreover, $c$ and $c'$ differ at some vertex $v$ if and only if the $\domino^+$-valid configuration is in some state from $S_0'$ at $v$. Therefore, it follows that $F_{G,f}$ satisfies formula of item two if and only if $\domino^+,s_0$ is a positive instance of the seeded domino problem for some element ${s_0\in S_0'}$ (hence a Turing reduction). Similarly, $F_{G,f}$ satisfies formula of item three if and only if $\domino^+,s_0$ is a positive instance of the recurring domino problem for some element ${s_0\in S_0'}$ (hence a Turing reduction).
\end{proof}

The recurring domino problem is $\Sigma_1^1$-hard on $\Z^2$ \cite{Harel_1985}, as well as the model checking of the corresponding \FOasy{} formula from Theorem~\ref{theo:dominofoca}. It is just an existential formula, but it crucially uses relation $\asympto$. We can actually also obtain $\Sigma_1^1$-hardness on $\Z^2$ with a pure FO formula with just one quantifier alternation, using preimage counting trickery to check finiteness of a set and a reduction from the recurring domino problem.

\begin{theorem}\label{theo:foanadur}
  The problem ${\cas{\phi,\Z^2}}$ is ${\Sigma_1^1}$-hard where $\phi$ is the following formula:
  \[\phi\isbydef\exists y,y\to y\wedge \forall y',\forall y_1, \forall y_2,\forall y_3, (y'\neq y\wedge y'\to y\wedge \bigwedge_i y_i\to y')\Rightarrow \bigvee_{i\neq j}y_i=y_j.\]
\end{theorem}

\newcommand\extn{\downarrow}
\newcommand\exts{\uparrow}
\newcommand\exte{\leftarrow}
\newcommand\extw{\rightarrow}
\newcommand\extne{\swarrow}
\newcommand\extnw{\searrow}
\newcommand\extse{\nwarrow}
\newcommand\extsw{\nearrow}
\newcommand\crne{C_{\texttt{ne}}}
\newcommand\crnw{C_{\texttt{nw}}}
\newcommand\crse{C_{\texttt{se}}}
\newcommand\crsw{C_{\texttt{sw}}}

\begin{proof}
  In plain English, formula $\phi$ requires the existence of a fixed point such that any of its pre-images (distinct from it) has at most two pre-images.  
  We proceed by reduction from the recurring domino problem (problem R1 of \cite{Harel_1985}). Consider an instance ${(\domino,s_0)}$ of this problem over states set $S$ and let us describe a CA local rule $f$ over states set $S'$, of radius $1$ and using $1$-capped multiset. The states from $S'$ are error states or states made of at most three layers, precisely:
  \[S' = S \cup S\times S_\square\cup S\times S_\square\times\{0,1\}\cup \{e_0,e_1\}\]
  where $e_0$ and $e_1$ are error states as usual, and $S_\square$ is used to code a particular set of configurations through domino constraints $\domino_\square$ (detailed below). The type of a state is either \textit{error}, or its number of layers (\textit{e.g.} $2$ for states from ${S\times S_\square}$). The key property of $\domino_\square$-valid configurations is that they either code a finite square zone with four corners where the interior can be distinguished from the exterior, or they contain at most one corner (possibly corresponding to having an infinite square with an infinite interior). More precisely, we define
  \[S_\square = \{N,S,E,W,E_d, W_d,I,I_d,I^d,\crne,\crnw,\crse,\crsw,\extn,\exts,\exte,\extw,\extne,\extnw,\extse,\extsw\}\]
  where $I$ and $I_d$ and $I^d$ are the interior states, and $\crne,\crnw,\crse,\crsw$ are corner states. $\domino_\square$ is the set of vertical and horizontal dominos appearing in this partial configuration:
  \[
    \begin{matrix}
      \extnw & \extn & \extn & \extn & \extn & \extn & \extn & \extn & \extne\\
      \extw & \extnw & \extn & \extn & \extn & \extn & \extn & \extne & \exte\\
      \extw & \extw & \crnw & N & N &N & \crne & \exte & \exte\\
      \extw & \extw & W & I & I^d & I_d & E_d & \exte & \exte\\
      \extw & \extw & W & I^d & I_d & I & E & \exte & \exte\\
      \extw & \extw & W_d & I_d & I & I & E & \exte & \exte\\
      \extw & \extw & \crsw & S & S & S & \crse & \exte & \exte\\
      \extw & \extsw & \exts & \exts & \exts & \exts & \exts & \extse & \exte\\
      \extsw & \exts & \exts & \exts & \exts & \exts & \exts & \exts & \extse\\
    \end{matrix}
  \]
The index \emph{d} in the above symbols mark a diagonal of the square (since we use only vertical and horizontal domino constraints, we need to mark staircase diagonals). It should be clear that a $\domino_\square$-valid configuration contains at most one occurrence of each corner state type, at most one horizontal segment of $N$ (resp. of $S$) and at most one vertical segment of $W_d$/$W$ (resp. $E$/$E_d$). Also a valid configuration with four corners contains finitely many occurrences of interior states. Moreover, it can be checked by a simple case analysis that if two corners are present, then actually four are present.
  
 The behavior of $f$ is essentially to let configurations of type $1$ unchanged, to project ${S\times S_\square\times\{0,1\}}$ onto ${S\times S_\square}$ and ${S\times S_\square}$ onto $S$. The local rule also verifies the following conditions (and generate an error with period two oscillations between $e_0$ and $e_1$ if one condition is violated):
\begin{itemize}
\item two neighboring states must have the same type,
\item a configuration using only states from $S$ must represents a $\domino$-valid configuration,
\item the $S_\square$ component of a configuration of type 2 or 3 must be a $\domino_\square$-valid configuration,
\item each occurrence of state $s_0$ in the $S$ component of states, must be at a vertex where the $\domino_\square$ component is in an interior state ($I$ or $I_d$ or $I^d$),
\item a configuration of type $3$ must be such that the $\{0,1\}$ component at any vertex is $0$, except when the $S_\square$ component represents a corner, in which case it can be either $0$ or $1$ (at this point the reader should see coming a preimage counting trick in order to count the number of corners).
\end{itemize}
With these conditions, the only possible fixed-points are type 1 configuration representing a $\domino$-valid configuration ; their pre-images must additionally contain a $S_\square$ component representing a $\domino_\square$-valid configuration ; and the number of possible pre-images of each such pre-image is $2^k$ where $k$ is the number of corners in the $S_\square$ component.

Therefore, a fixed-point such that each of its pre-image has at most $2$ pre-images (like specified in formula $\phi$) is a configuration of type 1 such that it is impossible to form a pre-image with two or more corners, \textit{i.e.} such that it is impossible to form a pre-image with finitely many interior states without generating an error state, and therefore it represents a $\domino$-valid configuration with infinitely many occurrences of $s_0$.
Conversely, it is clear that for each $\domino$-valid configuration with infinitely many occurrences of $s_0$ we can construct a fixed point with the above properties.
\end{proof}

\section{Perspectives}
\label{sec:open}
We see several interesting research directions inspired by the approach taken in this work.

First, we believe that the dependence of $\phi$ on the degree or the number of generators in Corollary~\ref{coro:onefo} is an artifact that can be removed with more work in the proof of Lemma~\ref{lem:notarithmso}.
The same proof techniques should also provide hardness result at any level of the analytical hierarchy.

Then, this corollary should be put into perspective with the Ballier-Stein conjecture \cite{Ballier_2018} saying that the domino problem on a f.g. group is decidable if and only if the group is virtually free. On one hand, it seems natural to ask whether the recurring domino problem (or its equivalent FO formula from Theorem~\ref{theo:dominofoca}) can play the role of formula $\phi'$ in Corollary~\ref{coro:onefo}. 
On the other hand, N. Pytheas Fogg pointed us simple examples of $4$-regular graphs having an $\infty$-grid as subgraph on which the domino problem is decidable. So formula $\phi$ in Corollary~\ref{coro:onefo} cannot be the FO formula expressing the existence of a fixed point (Turing-equivalent to the domino problem), and we wonder how simple such formula $\phi$ can be. Actually, we can ask a similar question for Corollary~\ref{coro:finitesatfoca}.

In general, we believe that the Turing degrees of FO-model checking problems for various concrete formulas is worth being investigated.  As mentioned above, the Turing degree of all such model checking problems for a fixed FO formula is independent of the choice of generators on f.g. groups, and we wonder how they change when changing the group among non virtually free groups. Injectivity of CA is a natural candidate that received little attention to our knowledge since the seminal result on $\Z^2$ \cite{Kari_injsurj}.

Finally, we believe that there exists a fixed CA rule $f$ for which  the FO-model checking problem is undecidable on graph $\Z^2$ (the rule is fixed, the formula is given as input). While we see the proof ingredient to obtain this specifically for $\Z^2$, we have no idea of whether it is always the case that undecidability of FO model checking for CA orbits can be obtained for a fixed CA rule on any f.g. which is not virtually free.

\section{Acknowledgment}
\label{sec:ack}

We thank anonymous referees for their feedback and their suggestions to improve the presentation. We also warmly thank N. Pytheas Fogg for their hints about domino problems on regular graphs and numerous stimulating discussions that inspired this work.

\bibliography{biblio.bib}

\end{document}